\documentclass[12pt]{article}

\usepackage{amsfonts,amsmath,amsthm,amsxtra}
\usepackage{wrapfig,caption,graphics,graphicx,xcolor,epstopdf, subfig}
\usepackage{tabularx, multirow, fancybox, lscape,listings,float,rotating,fancyhdr}
\usepackage{hyperref , cleveref, esvect, physics}
\usepackage{algorithm}
\usepackage{algpseudocode}
\usepackage{xfrac,nicefrac,textcomp}
\usepackage[sort&compress]{natbib}
\usepackage[us,12hr]{datetime}
\usepackage{matlab-prettifier}
\usepackage{cases}
\usepackage{enumitem}
\usepackage{booktabs}
\usepackage[affil-it]{authblk}
\usepackage{oubraces}
\numberwithin{equation}{section}
\usepackage{chngcntr}
\usepackage{soul}

\newcolumntype{C}[1]{>{\centering\arraybackslash}p{#1}}
\newcolumntype{L}[1]{>{\raggedright\arraybackslash}p{#1}}
\newcolumntype{M}[1]{>{\centering\let\newline\\\arraybackslash\hspace{0pt}}m{#1}}

\graphicspath{{Figures/}{images/}}

\newtheorem{theorem}{Theorem}[section]
\newtheorem*{remark}{Remark}

\newtheorem{corollary}[theorem]{Corollary}

\title{SCAR dynamics of adolescent substance use: peer influence, dropout, and bifurcation structure in a school-based model}
\author[a]{Tamantha Pizarro}%<tcpizarr@asu.edu>
\author[b]{Jinni Su}%jinnisu1@asu.edu
\author[c]{Yixuan He}%Yixuan.He@asu.edu>
\author[d]{Yun Kang\thanks{Corresponding author: Yun.Kang@asu.edu}}

\affil[a]{School of Complex Adaptive Systems, Arizona State University, Tempe, AZ, 85287, USA}
\affil[b]{Department of Psychology, Arizona State University, Tempe, AZ, 85287, USA}
\affil[c]{School of Mathematical and Natural Sciences
Arizona State University,Phoenix, AZ 85069, USA}
\affil[d]{College of Integrative Sciences and Arts, Arizona State University, Mesa, AZ, 85212, USA}

\begin{document}
\maketitle

\begin{abstract}
Adolescent substance use is shaped by peer influence, protective social environments, and school-related transitions such as disengagement and re-entry. To study these interacting mechanisms, we develop a four-compartment susceptible--casual--addicted--resistant (SCAR) model for adolescent substance use in a high-school setting. The population is divided into susceptible students, casual or experimental users, students with sustained, problematic, or substance-use-disorder (SUD)-level involvement, and a resistant class representing students embedded in comparatively strong anti-use or protective environments. The model incorporates peer-driven initiation, escalation from casual to problematic use, protective peer influence, substance-use-associated school disengagement, and partial return following rehabilitation or successful re-entry.

Using qualitative dynamical-systems analysis and one-parameter bifurcation diagrams, we obtain three main findings. First, the return parameter \(\phi\), representing the fraction of students in the problematic-use class who re-enter school after disengagement, separates the model into two qualitatively distinct regimes: when \(\phi=1\), the total population is conserved and biologically meaningful interior equilibria may exist, whereas when \(\phi<1\), problematic substance use induces net population loss, so positive equilibria of the scaled system need not correspond to true endemic equilibria in the original variables. Second, initiation and escalation are governed by distinct threshold mechanisms, showing that first use and progression to problematic use are dynamically different processes. Third, the model can exhibit multistability, including bistability between a substance-free equilibrium and a stable interior high-use equilibrium, so identical parameter values may produce different long-term outcomes depending on initial conditions.

These results have direct implications for school policy. They show that effective intervention must address not only initiation, but also escalation, protective social influence, and successful school re-entry after disengagement. In particular, the model supports a layered strategy combining universal prevention, early intervention for casual users, targeted support for students at risk of problematic use, recovery-supportive school environments, and strong school re-engagement pathways. Although motivated here by adolescent cannabis use, the framework also applies more broadly to alcohol use, nicotine use, vaping, and other adolescent substance-use settings in which peer influence and school-related transitions are central.
\end{abstract}

\section{Introduction}

Human behavior and social interaction play a central role in many public-health problems, including infectious-disease transmission, mental-health outcomes, and substance use \cite{bharti2021linking,holt2024social,strickland2014effects}. In the context of substance use, social influence through peers, perceived norms, and repeated interpersonal contact affects initiation, escalation, cessation, relapse, and recovery, and these mechanisms are especially important during adolescence, when peer relationships, perceived norms, and school environment strongly shape behavioral decisions \cite{strickland2014effects,whitesell2013familial}. Because adolescent substance use is driven in part by repeated interpersonal interactions, mathematical models that explicitly incorporate social influence can help identify mechanisms that promote or suppress use at the population level. This paper is a theoretical study: our goal is not to estimate effects from a specific dataset, but to use mathematical analysis and bifurcation structure to identify mechanisms that can generate qualitatively different school-level outcomes.

Adolescent substance use is also a developmental and educational concern. Repeated substance use during adolescence has been associated with impaired attention, memory, learning, coordination, academic performance, and increased risk of school disengagement and dropout \cite{national2017health,batalla2013structural,teensmaricdc}. Since adolescence is a period of ongoing cognitive, behavioral, and social development, substance exposure may affect not only immediate behavior but also longer-term educational and health trajectories \cite{national2017health,batalla2013structural}. Thus, adolescent substance use is not only a behavioral-health problem, but also a school-retention and developmental-health problem.

Among adolescent substance-use behaviors, cannabis provides a particularly important motivating application. Marijuana remains one of the most commonly used substances among teenagers in the United States \cite{teensmaricdc,johnston2023monitoring}. At the same time, changing legal and social attitudes toward adult recreational cannabis use have raised concern that adolescents may increasingly perceive marijuana as low-risk or socially acceptable \cite{legalmari,wilkinson2016formation,mennis2023recreational,nguyen2023short}. National survey data continue to document substantial cannabis use among secondary-school students, including both annual use and frequent use patterns \cite{johnston2023monitoring}. Because of its prevalence, changing social perception, and strong relevance to peer-influence processes, adolescent cannabis use provides a natural application through which to interpret a broader mathematical framework for adolescent substance use dynamics.

A high-school setting is particularly suitable for mathematical modeling because it combines repeated peer contact with structured turnover through enrollment, graduation, disengagement, and return. Students may also move through qualitatively different behavioral states, including non-use, casual or experimental use, more sustained or problematic use, and strong resistance to use. These features suggest that a school-based model should account simultaneously for peer-driven behavioral transitions and school-related turnover, especially when sustained substance use may affect school persistence and re-entry after intervention \cite{national2017health}. 

Mathematical models for substance-use dynamics have been developed for a variety of substances and populations using compartmental ordinary differential equations and related dynamical-systems frameworks \cite{abidemi2023optimal,andrawus2024unraveling,eguda2022analysis,ullah2024mathematical,ullah2024mathematical2,yusuf2014modelling}. In general, these models divide the population into behavioral classes such as susceptible, experimenting, addicted, rehabilitated, or recovered individuals, and then analyze how social contact, treatment, relapse, awareness, or intervention affect long-term prevalence. Abidemi \cite{abidemi2023optimal} studied substance involvement among students through a mathematical model with cost-effective optimal control. Andrawus et al.~\cite{andrawus2024unraveling} emphasized the importance of early awareness strategies in addiction dynamics. Eguda et al.~\cite{eguda2022analysis} focused on equilibrium structure and stability in a youth addiction model. Ullah et al.~\cite{ullah2024mathematical,ullah2024mathematical2} developed marijuana-use models incorporating sensitivity analysis and control strategies. Yusuf and Benyah \cite{yusuf2014modelling} provided an important foundation by formulating a compartmental framework for adult marijuana use and analyzing its dynamics.

Despite these advances, comparatively fewer models are tailored to adolescent or school-based populations, where peer effects are especially strong and where educational transitions such as dropout, graduation, and return directly alter the effective population under study. This distinction is not merely contextual; it changes the demographic and behavioral mechanisms that must be represented mathematically. In a high-school setting, repeated peer contact, protective social influence, substance-use-associated disengagement, and partial re-entry after intervention are all central features of the system. Here, “school disengagement” matters because the model tracks the population participating in the school contact structure; students who leave school are treated as leaving that structured interaction environment, even though broader off-campus peer interactions may still exist.
These considerations motivate the development of a mathematically tractable model specifically adapted to adolescent substance use in a school-based environment.

In this paper, we formulate and analyze a four-compartment SCAR model for adolescent substance use dynamics in a high-school setting, with application to adolescent cannabis use. The four compartments are intended to distinguish non-use with varying degrees of vulnerability or protection, early-stage use, and more severe problematic use, rather than to represent every possible behavioral subtype.
The population is divided into susceptible students \(S(t)\), casual or experimental users \(C(t)\), students with sustained, problematic, or substance-use-disorder (SUD)-level use \(A(t)\), and a resistant class \(R(t)\). The model incorporates peer-driven initiation, progression from casual use to problematic use, protective influence, substance-use-associated dropout, and partial return to the school population following rehabilitation or successful re-entry. In this way, the model combines behavioral contagion with school-population turnover while remaining sufficiently simple for rigorous qualitative analysis.

The remainder of the paper is organized as follows. In Section 2, we derive the SCAR model and explain its interpretation in the high-school setting. Section 3 contains the mathematical analysis, including positivity and boundedness, the scaled system, subsystem dynamics, threshold conditions, and equilibrium analysis. Section 4 presents the bifurcation diagrams and their interpretation. Section 5 concludes with a discussion of the main implications, limitations, and directions for future work.
\section{Model Derivation}

We formulate our model by adapting the compartmental structure proposed by Yusuf et al.~\cite{Yusuf2014}, who studied marijuana use dynamics in an adult population. Here we reinterpret that SCAR framework in a high-school setting and apply it to adolescent substance use, with adolescent cannabis use serving as a motivating application. This setting is especially relevant because adolescent substance use is shaped not only by peer influence, but also by school-specific processes such as enrollment, graduation, disengagement, and re-entry.

We divide the school population into four mutually exclusive classes. Let \(S(t)\) denote students who are currently not using substances but remain behaviorally susceptible to initiation under sufficient peer or environmental pressure. Let \(C(t)\) denote students engaged in casual, experimental, or intermittent use. Let \(A(t)\) denote students with sustained, problematic, or substance-use-disorder (SUD)-level use. Let \(R(t)\) denote students who are currently non-using and embedded in comparatively strong protective environments or anti-use norms, so that their risk of initiation is lower than that of the susceptible class. Thus, the distinction between \(S\) and \(R\) is not simply abstinent behavior versus use, but comparatively vulnerable non-use versus comparatively protected non-use.

The model assumes continuous turnover in the high-school population, with students entering and exiting at a constant per-capita rate \(\mu\). Here exit represents graduation or other school departure unrelated to natural death, which is negligible on the time scale of interest. New students enter either the susceptible class \(S\) or the resistant class \(R\): a fraction \(\rho\) enters directly into \(R\), while the remaining fraction \(1-\rho\) enters \(S\). This reflects heterogeneity in baseline attitudes toward substance use at school entry, since some students arrive with stronger anti-use norms because of family, cultural, religious, or other protective influences \cite{Hawkins1992,Hemphill2011,Hawkins1992b,SAMHSA2024}. Once in the resistant class, students remain there until leaving the school population. This is a baseline simplifying assumption rather than a claim that all entering students are substance-naive. It represents a school-entry classification into comparatively susceptible versus comparatively protected non-using states. Allowing direct entry into \(C\) or \(A\) would be a natural extension of the model in future work.

Transitions between behavioral classes are modeled using frequency-dependent terms of the form \(\beta XY/N\), which are standard in compartmental models of behavioral transmission and substance-use dynamics \cite{Binuyo2021,Ali2024}. In the present setting, these terms represent effective school-based peer influence. Some of these influences increase substance-use risk, while others are protective.

Susceptible students may initiate substance use through contact with either casual users or problematic users, represented by \(\beta_{CS}CS/N\) and \(\beta_{AS}AS/N\), respectively. These terms encode risk-promoting peer influence. By contrast, the terms \(\beta_{RS}RS/N\), \(\beta_{RC}RC/N\), and \(\beta_{RA}AR/N\) encode protective influence, meaning social processes that reduce initiation risk or promote de-escalation, such as prosocial peers, mentoring, anti-use norms, family support, or recovery-supportive school environments.

Students in the casual-use class may either return to susceptibility through protective influence from resistant students, represented by \(\beta_{RC}RC/N\), or progress to problematic use through contact with students in the problematic-use class, represented by \(\beta_{AC}AC/N\). These competing pathways describe the balance between de-escalation and escalation in adolescent substance use \cite{Schuler2019,Watts2024,Woodward2023}. Students in the problematic-use class may de-escalate into the casual-use class through protective influence from resistant students, represented by \(\beta_{RA}AR/N\), reflecting the role of recovery-supportive school and peer environments \cite{Stritzel2021,Hawkins1992,Woodward2023}.

We further assume that students in the problematic-use class may disengage from school at rate \(\delta\). A fraction \(\phi\) later returns to the school population and re-enters the susceptible class, whereas the remaining fraction \(1-\phi\) leaves the modeled school population permanently. Thus \(\delta\) measures school disengagement associated with problematic substance use, while \(\phi\) measures the extent to which that disengagement is reversible through rehabilitation or successful re-entry \cite{Hawkins1992,Stritzel2021,Do2014,SAMHSA2024}. Students returning after disengagement are assigned to \(S\) rather than \(C\), \(A\), or \(R\) as a baseline assumption representing re-entry into the school population without assuming either immediate relapse or immediate durable protection. This choice keeps the re-entry mechanism mathematically simple while allowing subsequent risk or protection to be shaped by the school environment.

Finally, students are not assumed to move directly from the casual-use or problematic-use classes into the resistant class. Instead, the model interprets \(R\) as a durable protected state that is reached gradually through pathways such as \(A \to C \to S \to R\), reflecting re-integration into stable anti-use environments rather than immediate post-recovery resistance.
Under these assumptions, the SCAR model takes the form
\begin{subequations}\label{eq:fullmod}
\begin{align}
\frac{dS}{dt} &=
    \mu(1-\rho)N
    - \frac{\beta_{CS} C S}{N}
    - \frac{\beta_{AS} A S}{N}
    - \mu S
    - \frac{\beta_{RS} R S}{N}
    + \frac{\beta_{RC} R C}{N}
    + \phi\delta A, \label{eq:S}\\[1ex]
\frac{dC}{dt} &=
    \frac{\beta_{CS} C S}{N}
    + \frac{\beta_{AS} A S}{N}
    - \mu C
    - \frac{\beta_{RC} R C}{N}
    - \frac{\beta_{AC} A C}{N}
    + \frac{\beta_{RA} A R}{N}, \label{eq:C}\\[1ex]
\frac{dA}{dt} &=
    \frac{\beta_{AC} A C}{N}
    - \mu A
    - \delta A
    - \frac{\beta_{RA} A R}{N}, \label{eq:A}\\[1ex]
\frac{dR}{dt} &=
    \mu\rho N
    + \frac{\beta_{RS} R S}{N}
    - \mu R. \label{eq:R}
\end{align}
\end{subequations}

Here
\[
N(t)=S(t)+C(t)+A(t)+R(t),
\]
since the four compartments are mutually exclusive and collectively exhaustive. Summing \eqref{eq:S}--\eqref{eq:R} yields
\[
\frac{dN}{dt}=-(1-\phi)\delta A(t).
\]
Therefore, the total high-school population is constant if and only if \(\phi=1\), and it decreases whenever \(A(t)>0\) if \(\phi<1\).

From a mechanistic standpoint, system~\eqref{eq:fullmod} has a triangular structure: the proportion dynamics evolve independently of \(N\), while \(N\) evolves according to a scalar equation driven by the problematic-use class. This separation is central to the distinct qualitative behaviors of the model in the cases \(\phi=1\) and \(\phi<1\), which will be analyzed in the next section.

Table~\ref{tab:parameters} summarizes all variables and parameters, and Figure~\ref{diagram} illustrates the overall flow structure of the SCAR model.

\begin{table}[ht]
    \centering
    \caption{Description of variables and parameters for the SCAR model~\eqref{eq:fullmod}.}
    \begin{tabular}{M{4cm}|L{10cm}|c}
        \toprule
        \textbf{Symbol} & \textbf{Description} & \textbf{Units} \\
        \midrule
        $S$ & Susceptible (non-using) students & Individuals \\[1ex]
        $C$ & Casual, experimental, or intermittent substance users & Individuals \\[1ex]
        $A$ & Students with sustained, problematic, or SUD-level substance use & Individuals \\[1ex]
        $R$ & Resistant/protected non-using students with strong anti-use norms & Individuals \\[1ex]
        $N$ & Total school population, $N=S+C+A+R$ & Individuals \\[1ex]
        $\mu$ & Per-capita rate of entry and exit from the school (turnover rate) & $\text{time}^{-1}$ \\[1ex]
        $\rho$ & Fraction of incoming students with strong no-use norms (entering directly into $R$) & -- \\[1ex]
        $\beta_{CS}$ & Peer-influence rate: contact between $C$ and $S$ leading $S \to C$ & $\text{time}^{-1}$ \\[1ex]
        $\beta_{AS}$ & Peer-influence rate: contact between $A$ and $S$ leading $S \to C$ & $\text{time}^{-1}$ \\[1ex]
        $\beta_{RS}$ & Protective influence rate: contact between $R$ and $S$ leading $S \to R$ & $\text{time}^{-1}$ \\[1ex]
        $\beta_{RC}$ & Protective influence rate: contact between $R$ and $C$ leading $C \to S$ & $\text{time}^{-1}$ \\[1ex]
        $\beta_{AC}$ & Escalation rate: contact between $A$ and $C$ leading $C \to A$ & $\text{time}^{-1}$ \\[1ex]
        $\beta_{RA}$ & Protective de-escalation rate: contact between $R$ and $A$ leading $A \to C$ & $\text{time}^{-1}$ \\[1ex]
        $\delta$ & School disengagement/dropout rate associated with problematic substance use & $\text{time}^{-1}$ \\[1ex]
        $\phi$ & Fraction of disengaged students who later re-enter school & -- \\[1ex]
        \bottomrule
    \end{tabular}
    \label{tab:parameters}
\end{table}

\begin{figure}[ht]
    \centering
    \includegraphics[width=0.85\textwidth]{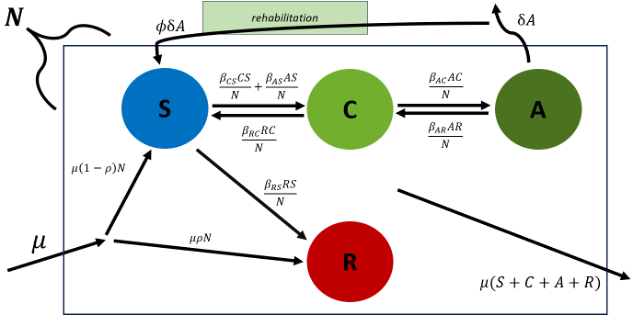}
    \caption{Flow diagram illustrating all transitions in the SCAR model~\eqref{eq:fullmod}, including peer-influence routes, protective influences, de-escalation pathways, and dropout/re-entry dynamics.}
    \label{diagram}
\end{figure}
\section{Mathematical Analysis}
In this section we establish that Model~\eqref{eq:fullmod} is mathematically well posed and then analyze the reduced proportion dynamics. We first prove positivity and boundedness of solutions. We then derive the scaled system, identify the distinct dynamical roles of the return parameter \(\phi\), analyze the substance-free and SUD-free subsystems, derive invasion thresholds for casual and problematic use, and characterize the interior equilibrium structure of the full reduced system when \(\phi=1\).

\begin{theorem}[Positive invariance and boundedness]\label{posbound}
Let \((S(0),C(0),A(0),R(0))\in \mathbb{R}_{+}^{4}\). Then the solution
\((S(t),C(t),A(t),R(t))\) of system~\eqref{eq:fullmod} exists for all \(t\ge 0\),
remains in \(\mathbb{R}_{+}^{4}\), and satisfies
\[
0 \le S(t),C(t),A(t),R(t) \le N(t)
\quad\text{for all } t\ge 0,
\]
where \(N(t)=S(t)+C(t)+A(t)+R(t)\). Moreover,
\[
\frac{dN}{dt} = -(1-\phi)\delta A(t)\le 0,
\]
so \(N(t)\) is nonincreasing. In particular,
\[
0 < N(t) \le N(0)\qquad \forall t\ge 0,
\]
and therefore all solutions remain bounded in the compact region
\[
\Omega = \{(S,C,A,R)\in \mathbb{R}_{+}^{4} : 0\le S+C+A+R\le N(0)\}.
\]
Thus \(\Omega\) is a positively invariant and bounded region for Model~\eqref{eq:fullmod}.
\end{theorem}

\noindent\textbf{Interpretation.}
Theorem~\ref{posbound} guarantees that the SCAR model remains biologically meaningful for all forward time: the sizes of the susceptible, casual-use, problematic-use, and resistant classes remain nonnegative, and the total modeled school population remains finite.

%\noindent\textbf{The scaled SCAR system}\\

To obtain a reduced system describing the evolution of proportions, we introduce
\[
x=\frac{S}{N},\qquad y=\frac{C}{N},\qquad z=\frac{A}{N},\qquad w=\frac{R}{N},
\qquad x+y+z+w=1.
\]
Because \(N(t)>0\) for all \(t\ge0\), this change of variables is well defined. Differentiating \(x=S/N\) and similarly for \(y,z,w\), and using
\[
N'=-(1-\phi)\delta A=-(1-\phi)\delta Nz,
\]
we obtain the proportion system
\begin{subequations}\label{eq:scaled}
\begin{align}
x' &= \mu(1-\rho)
  -x\mu - x(1-\phi)\delta z
  -\beta_{CS}xy -\beta_{AS}xz
  -\beta_{RS}xw + \beta_{RC}wy + \phi\delta z,\\[1ex]
y' &= \beta_{CS}xy + \beta_{AS}xz
  -y\mu -y(1-\phi)\delta z
  -\beta_{RC}wy -\beta_{AC}zy +\beta_{RA}zw,\\[1ex]
z' &= \beta_{AC}zy - z\mu -z\delta -z(1-\phi)\delta z -\beta_{RA}zw,\\[1ex]
w' &= \mu\rho -w\mu -w(1-\phi)\delta z +\beta_{RS}xw.
\end{align}
\end{subequations}

Since \(x=1-y-z-w\), the scaled system reduces to the \(yzw\)-system
\begin{subequations}\label{eq:scaled_reduced}
\begin{align}
y' &=
y\left[\beta_{CS}(1-y-z-w)-\mu-(1-\phi)\delta z-\beta_{RC}w-\beta_{AC}z-\beta_{AS}z\right]
+z\left[\beta_{AS}(1-z-w)+\beta_{RA}w\right],\\
z' &= z\left[\beta_{AC}y - \mu -\delta -(1-\phi)\delta z -\beta_{RA}w\right],\\[1ex]
w' &= \mu\rho -w\left[\mu +(1-\phi)\delta z -\beta_{RS}(1-y-z-w)\right].
\end{align}
\end{subequations}

\begin{remark}
The scaled formulation separates changes in behavioral composition from changes in total school population. This distinction is essential because, when \(\phi<1\), the scaled system may admit positive steady states in proportions that do not correspond to genuine endemic equilibria in the original variables.
\end{remark}

%\nonindent\textbf{The structural role of the return parameter \texorpdfstring{$\phi$}{phi}}

The parameter \(\phi\) determines whether school disengagement associated with problematic use is fully reversible at the population level. Since
\[
\frac{dN}{dt}=-(1-\phi)\delta A(t),
\]
the sign of \(1-\phi\) separates the system into two structurally different regimes.

\noindent\textbf{Case \(\phi=1\): conserved-population regime}

When \(\phi=1\),
\[
\frac{dN}{dt}=0,
\]
so \(N(t)\equiv N(0)\). In this regime, the SCAR model behaves like a classical compartmental system with conserved total population. The reduced \(yzw\)-system becomes
\begin{subequations}\label{eq:scaledphi10}
\begin{align}
y' &=
y\left[\beta_{CS}(1-y-z-w)-\mu-\beta_{RC}w-(\beta_{AC}+\beta_{AS})z\right]
+z\left[\beta_{AS}(1-z-w)+\beta_{RA}w\right],\\
z' &= z\left[\beta_{AC}y-\mu-\delta-\beta_{RA}w\right],\\
w' &= \mu\rho-w\left[\mu-\beta_{RS}(1-y-z-w)\right].
\end{align}
\end{subequations}
In particular, biologically meaningful interior equilibria may exist in the original variables.

\noindent\textbf{Case \(\phi<1\): population-loss regime}

When \(0\le \phi<1\),\qquad
\(\frac{dN}{dt}=-(1-\phi)\delta A(t)<0
\qquad\text{whenever } A(t)>0.\)
 Thus problematic substance use contributes directly to loss from the modeled school population.

This has an immediate consequence: the original SCAR system cannot admit a biologically meaningful equilibrium with \(A^*>0\). Indeed, if \((S^*,C^*,A^*,R^*)\) were an equilibrium with \(A^*>0\), then
\[
N'=-(1-\phi)\delta A^*<0,
\]
contradicting the equilibrium condition \(N'=0\). Therefore, for \(\phi<1\), any positive equilibrium in the scaled system should be interpreted only as a \emph{shape equilibrium} in proportions, not as a true endemic equilibrium in the original variables.

\begin{remark}
The distinction between \(\phi=1\) and \(\phi<1\) is one of the main structural features of the model. When \(\phi=1\), the system may support true long-term coexistence of all classes. When \(\phi<1\), problematic use induces net school-population loss, so only SUD-free equilibria remain biologically meaningful in the original variables.
\end{remark}
\subsection{Dynamics of the substance-free and SUD-free subsystems}

We begin with the scaled substance-free subsystem, corresponding to the case \(C=A=0\), so that only the susceptible and resistant classes remain. 
\begin{equation}\label{eq:scaledphi1yz0}
w'=\mu\rho-w\left[\mu-\beta_{RS}(1-w)\right].
\end{equation}

\begin{theorem}[Global dynamics of the scaled substance-free subsystem]\label{thm:drugfree_wf}
Consider \eqref{eq:scaledphi1yz0} with \(\mu>0\), \(\rho\in[0,1]\), and \(\beta_{RS}>0\). Then \eqref{eq:scaledphi1yz0} has a unique equilibrium
\begin{equation}\label{eq:wf_explicit}
w_f=\frac{(\beta_{RS}-\mu)+\sqrt{(\beta_{RS}-\mu)^2+4\rho\mu\beta_{RS}}}{2\beta_{RS}},
\end{equation}
and \(w_f\) is globally asymptotically stable in \([0,1]\).
\end{theorem}

\noindent\textbf{Interpretation and policy implications.}
Theorem~\ref{thm:drugfree_wf} shows that, in the absence of the casual-use and problematic-use classes, the school converges to a unique substance-free equilibrium determined by the balance among turnover \(\mu\), direct entry into the resistant class \(\rho\), and protective influence \(\beta_{RS}\). The equilibrium value \(w_f\) therefore measures the baseline level of resistant or protected students sustained by the school environment when active substance use is absent. Larger values of \(\rho\) increase the proportion of students entering school with strong no-use norms, while larger values of \(\beta_{RS}\) strengthen movement from susceptibility into resistance through protective peer influence. Conversely, higher turnover \(\mu\) tends to weaken the ability of the school to maintain a strongly protected composition. Biologically, this is consistent with evidence that anti-use norms, prosocial peer environments, and family protection reduce adolescent initiation risk \cite{Hawkins1992,Hemphill2011,Hawkins1992b,SAMHSA2024,Berkowitz2004}. From a policy perspective, this theorem supports universal prevention strategies that increase \(\rho\) and \(\beta_{RS}\), such as orientation-based prevention, school-wide norm setting, peer leadership, and family-school protective coordination.\\

Observe that the set \(z=0\), corresponding to the absence of the problematic-use or SUD-level class \(A\), is forward invariant in the scaled system. The resulting SUD-free subsystem is
\begin{subequations}\label{eq:scaledphi1z0}
\begin{align}
y' &=y\left[\beta_{CS}(1-y-w)-\mu-\beta_{RC}w\right],\\
w' &= \mu\rho-w\left[\mu-\beta_{RS}(1-y-w)\right]
\end{align}
\end{subequations}
forms a lower-dimensional invariant manifold of the full system.

\begin{theorem}[Dynamics of the SUD-free subsystem]\label{thm:yw}
Assume \(\beta_{CS},\beta_{RS},\beta_{RC},\mu>0\) and \(\rho\in(0,1)\). The subsystem \eqref{eq:scaledphi1z0} is positively invariant in
\[
\Omega=\{(y,w)\in\mathbb{R}_+^2:\ y+w\le 1\}.
\]
Moreover:

\noindent\textbf{(i) Boundary equilibrium.}
The subsystem admits a unique boundary equilibrium
\[
E_f=(0,w_f),
\]
which is locally asymptotically stable if and only if
\begin{equation}\label{eq:wthr}
w_f=\frac{(\beta_{RS}-\mu)+\sqrt{(\beta_{RS}-\mu)^2+4\rho\mu\beta_{RS}}}{2\beta_{RS}}>\frac{\beta_{CS}-\mu}{\beta_{CS}+\beta_{RC}}=:w_{\mathrm{thr}}.
\end{equation}

\noindent\textbf{(ii) Interior equilibria.}
Let
\(\Delta=\mu^2(\beta_{CS}-\beta_{RS})^2-4\mu\rho\,\beta_{RS}\beta_{RC}\beta_{CS}.\)

If \(\Delta\ge 0\), define
\(
w_c^\pm=\frac{\mu(\beta_{CS}-\beta_{RS})\pm\sqrt{\Delta}}{2\beta_{RS}\beta_{RC}},
\qquad
y_c^\pm=\frac{\beta_{CS}-\mu-(\beta_{CS}+\beta_{RC})w_c^\pm}{\beta_{CS}}.
\)
Then:
\begin{enumerate}
    \item[(a)] If \(0<w_c^-<\frac{\beta_{CS}-\mu}{\beta_{CS}+\beta_{RC}}=:w_{\mathrm{thr}},\) then \(E_c^-=(y_c^-,w_c^-)\) exists 
    in which case \(E_c^-\) is locally asymptotically stable.

    \item[(b)] If \(0<w_c^+<\frac{\beta_{CS}-\mu}{\beta_{CS}+\beta_{RC}}=:w_{\mathrm{thr}},\) then \(E_c^+=(y_c^+,w_c^+)\) exists 
    in which case \(E_c^+\) is a saddle equilibrium.
\end{enumerate}

\noindent\textbf{(iii) Absence of periodic orbits.}
The subsystem admits no periodic orbits in \(\Omega^\circ\) which is the interior of \(\Omega\).

\noindent\textbf{(iv) Global phase portrait.}
\begin{enumerate}
    \item[(a)] If \(0<w_c^-\leq w_c^+<\frac{\beta_{CS}-\mu}{\beta_{CS}+\beta_{RC}}=:w_{\mathrm{thr}},\) then
    both \(E_c^-\) and \(E_c^+\) exist, and the subsystem exhibits bistability between \(E_c^-\) and \(E_f\).
    \item[(b)] If \(0<w_c^-<w_{\mathrm{thr}}<w_c^+\), then \(E_c^-\) is globally asymptotically stable in \(\Omega^\circ\).
    \item[(c)] If \(w_{\mathrm{thr}}<w_c^-\), then no interior equilibrium exists and \(E_f\) is globally asymptotically stable.
\end{enumerate}
\end{theorem}

\noindent\textbf{Interpretation and policy implications.}
Theorem~\ref{thm:yw} shows that even when the problematic-use class is absent, the school may converge to the substance-free equilibrium, to a stable casual-use equilibrium, or to a bistable regime in which both outcomes are possible. The key balance is between peer-driven initiation into casual use, governed primarily by \(\beta_{CS}\), and protective return from casual use to non-use, governed by \(\beta_{RC}\), together with the background level of resistance determined by \(w_f\). The condition
\[
w_f>\frac{\beta_{CS}-\mu}{\beta_{CS}+\beta_{RC}}=:w_{\mathrm{thr}}
\]
means that the school’s protective structure is strong enough to suppress sustained casual use. In contrast, if \(w_f<w_{\mathrm{thr}}\), the school becomes vulnerable to long-term casual-use dynamics. Biologically, this supports the idea that early-use clusters, weakening of protective peer culture, or delayed intervention may move a school across a threshold into the basin of attraction of the stable casual-use equilibrium. From a policy perspective, this theorem supports early intervention before casual use becomes normalized. In model terms, schools should seek to reduce \(\beta_{CS}\) and increase \(\beta_{RC}\) and \(\beta_{RS}\), for example through peer-norm correction, mentoring, rapid response to early-use clusters, and strong prosocial peer structures.\\

\noindent\textbf{An illustrative example of bistability in the \((y,w)\) subsystem.}
Consider
\[
\beta_{CS}=4.7237,\quad
\beta_{RS}=2.0329,\quad
\beta_{RC}=2.1955,\quad
\mu=0.8054,\quad
\rho=0.02165.
\]
The boundary equilibrium is
\[
E_f=(0,w_f)\approx(0,0.6177),
\]
with eigenvalues
\[
\lambda\approx -1.2839,\quad -0.3557,
\]
so \(E_f\) is locally asymptotically stable.

The interior equilibrium condition reduces to
\[
\beta_{RS}\beta_{RC}w^2-\mu(\beta_{CS}-\beta_{RS})w+\mu\rho\,\beta_{CS}=0,
\]
which has two positive roots
\[
w_c^- \approx 0.04156,\qquad
w_c^+ \approx 0.4440.
\]
The corresponding \(y\)-values are
\[
y_c^- \approx 0.7686,\qquad
y_c^+ \approx 0.1791.
\]
Thus the two interior equilibria are
\[
E_c^- \approx (0.7686,\,0.04156),\qquad
E_c^+ \approx (0.1791,\,0.4440).
\] where \(E_c^-\)is locally asymptotically stable and \(E_c^+\) is a saddle. Therefore, the subsystem is bistable, with \(E_f\) and \(E_c^-\) as two competing attractors separated by the saddle \(E_c^+\).

\medskip
\noindent\textbf{Implications.}
This example shows that two schools with the same parameter values may still converge to different long-term outcomes, depending on their initial behavioral composition. In particular, the balance among initiation \((\beta_{CS})\), protective return \((\beta_{RC})\), resistance formation \((\beta_{RS})\), turnover \((\mu)\), and protected entry \((\rho)\) can generate two competing attractors. Biologically, this means that early-use clusters, weakening of protective peer culture, or delayed intervention may shift the school from a no-use state to a substance-use state. From a policy perspective, this supports rapid response, early screening, and prevention during key transition periods such as school entry, since later intervention may be more difficult once the system moves into the basin of attraction of the substance-use equilibrium.
\subsection{Local stability and invasion thresholds in the full reduced system}

We now return to the full reduced \(yzw\)-system \eqref{eq:scaled_reduced}. The subsystem analysis suggests two distinct invasion thresholds: one for casual use invading the substance-free equilibrium, and one for problematic use invading the SUD-free equilibrium.

\begin{theorem}[Local stability and invasion thresholds in the reduced \(yzw\) system]\label{thm:yzw_stability}
Define the \emph{substance-use invasion threshold at the substance-free equilibrium \(E_f=(0,0,w_f)\)} by
\begin{equation}\label{eq:Ry}
\mathcal{R}_y:=\frac{\beta_{CS}(1-w_f)}{\mu+\beta_{RC}w_f},
\end{equation}
and define the \emph{problematic-use invasion threshold at the SUD-free equilibrium \(E_c^-=(y_c^-,0,w_c^-)\)} by
\begin{equation}\label{eq:Rz}
\mathcal{R}_z^{(c)}:=\frac{\beta_{AC}y_c^-}{\mu+\delta+\beta_{RA}w_c^-}.
\end{equation}
Then:
\begin{enumerate}
    \item If \(\mathcal{R}_y<1\), equivalently \(w_f>w_{\mathrm{thr}}\), then the substance-free equilibrium \(E_f=(0,0,w_f)\) is locally asymptotically stable in the full \(yzw\) system.
    \item If \(0<w_c^-<w_{\mathrm{thr}}\) and \(\mathcal{R}_z^{(c)}<1\), then the SUD-free equilibrium \(E_c^-=(y_c^-,0,w_c^-)\) is locally asymptotically stable in the full \(yzw\) system.
\end{enumerate}
\end{theorem}

\begin{theorem}[A permanence criterion for the reduced \(yzw\) system]\label{thm:yzw_persistence}
Assume
\[
0<w_c^-<w_{\mathrm{thr}}<w_c^+
\qquad\text{and}\qquad
\mathcal{R}_z^{(c)}>1.
\]
Then the problematic-use class \(z\) is uniformly persistent in the reduced \(yzw\) system, thus the reduced \(yzw\) system is permanent.
\end{theorem}

\noindent\textbf{Theoretical implications.}
Theorem~\ref{thm:yzw_stability}-\ref{thm:yzw_persistence} identifies two distinct threshold mechanisms in adolescent substance-use dynamics. The threshold
\[
\mathcal{R}_y=\frac{\beta_{CS}(1-w_f)}{\mu+\beta_{RC}w_f}
\]
determines whether the casual-use class can invade a largely substance-free school environment. It increases with initiation pressure \(\beta_{CS}\) and decreases with turnover \(\mu\) and protective return \(\beta_{RC}w_f\). The threshold
\[
\mathcal{R}_z^{(c)}=\frac{\beta_{AC}y_c^-}{\mu+\delta+\beta_{RA}w_c^-}
\]
determines whether the problematic-use class can invade the SUD-free equilibrium once casual use is already established. It increases with escalation pressure \(\beta_{AC}\) and decreases with turnover \(\mu\), disengagement \(\delta\), and protective de-escalation \(\beta_{RA}w_c^-\). Thus the model separates initiation and escalation into related but distinct dynamical mechanisms. Biologically, this indicates that preventing first use and preventing progression to more severe use are different intervention problems \cite{Eisenberg2014,Schuler2019,Woodward2023,Watts2024,Loke2013,Stritzel2021,Hawkins1992}. From a policy perspective, \(\mathcal R_y>1\) points to the need for stronger universal prevention, whereas \(\mathcal R_z^{(c)}>1\) indicates the need for targeted intervention among already-using students.

\medskip
\noindent\textbf{A numerical example of full-system bistability for \(\phi=1\).}
We now present a numerical example illustrating bistability in the full scaled system. This example is intended as numerical evidence for the full \(yzw\) system and should not be interpreted as a direct corollary of Theorem~\ref{thm:yzw_persistence}.

For
\[
\beta_{CS}=3.0237,\quad
\beta_{AS}=1.6502,\quad
\beta_{RC}=2.6167,\quad
\beta_{AC}=3.3498,\quad
\beta_{RA}=1.9221,\quad
\beta_{RS}=0.8689,
\]
\[
\mu=0.4926,\quad
\delta=0.2383,\quad
\rho=0.07859,
\]
the full scaled system admits a locally asymptotically stable substance-free equilibrium
\[
E_f\approx(0.4811,\,0,\,0,\,0.5189),
\]
and a locally asymptotically stable interior equilibrium
\[
E^*\approx(0.3604,\,0.3420,\,0.0820,\,0.2157).
\]
Numerical simulations indicate that trajectories with sufficiently small initial substance-use levels converge to \(E_f\), whereas trajectories with sufficiently large initial substance-use levels converge to \(E^*\). Thus the full scaled system exhibits bistability between a substance-free equilibrium and a stable interior equilibrium.

For this parameter set,
\[
w_f \approx 0.5189,\qquad
w_{\mathrm{thr}} \approx 0.4487,\qquad
w_c^- \approx 0.1786,\qquad
w_c^+ \approx 0.2882,
\]
so
\[
0<w_c^-<w_c^+<w_{\mathrm{thr}}.
\]
Accordingly, this example lies outside the sufficient hypotheses of Theorem~\ref{thm:yzw_persistence} and should be interpreted as numerical evidence for bistability beyond the theorem’s sufficient parameter regime.

\subsection{Interior equilibria of the full reduced system when \texorpdfstring{$\phi=1$}{phi=1}}

We now specialize to the conserved-population regime \(\phi=1\), where the reduced system is \eqref{eq:scaledphi10}. Define
\begin{equation}\label{eq:wzplus}
w_z^+
=
\frac{
\left(1-\frac{\mu+\delta}{\beta_{AC}}-\frac{\mu}{\beta_{RS}}\right)
+
\sqrt{
\left(1-\frac{\mu+\delta}{\beta_{AC}}-\frac{\mu}{\beta_{RS}}\right)^2
+
4\left(1+\frac{\beta_{RA}}{\beta_{AC}}\right)\frac{\mu\rho}{\beta_{RS}}
}
}{
2\left(1+\frac{\beta_{RA}}{\beta_{AC}}\right)
}.
\end{equation}

\begin{theorem}[Interior equilibrium for the \(yzw\) system when \(\phi=1\)]\label{thm:phi1_interior_yzw}
Assume \(\phi=1\). Any interior equilibrium \(E^*=(y^*,z^*,w^*)\in\Omega\) with \(y^*,z^*,w^*>0\) satisfies
\begin{equation}\label{eq:phi1_y_of_w}
y^*=\frac{\mu+\delta+\beta_{RA}w^*}{\beta_{AC}},
\qquad
x^*=\frac{\mu(w^*-\rho)}{\beta_{RS}w^*},
\qquad
z^*=1-w^*-y^*-x^*.
\end{equation}
The admissible interval for \(w^*\) is
\[
w^*\in\mathcal W=(\rho,\min\{1,w_z^+\}),
\]
and \(w^*\) must satisfy the scalar equation
\begin{equation}\label{eq:w_scalar_phi1_clean}
\frac{
[(\beta_{RS}-\mu)w^*+\mu\rho-\beta_{RS}(w^*)^2]
[(\beta_{AS}\mu-\mu\beta_{RS}-\delta\beta_{RS})w^*-\beta_{AS}\mu\rho]
}
{
\beta_{RS}w^*
\big[\beta_{RC}\beta_{RS}(w^*)^2
+(\beta_{AS}-\beta_{CS})\mu(w^*-\rho)
-\delta\beta_{RS}w^*\big]
}
=
\frac{\mu+\delta+\beta_{RA}w^*}{\beta_{AC}}.
\end{equation}
Moreover:
\begin{enumerate}
    \item If \(\beta_{AC}<\mu+\delta\), then the \(yzw\) system has no interior equilibrium and \(\lim_{t\to\infty}z(t)=0\).
    \item If the scalar equation \eqref{eq:w_scalar_phi1_clean} admits a root \(w^*\in\mathcal W\), then the system has an interior equilibrium. In particular, under the hypotheses of Theorem~\ref{thm:yzw_persistence}, at least one interior equilibrium exists.
\end{enumerate}
\end{theorem}

\noindent\textbf{Interpretation and policy implications.}
Theorem~\ref{thm:phi1_interior_yzw} characterizes when all four behavioral classes can coexist in the long run under the conserved-population regime \(\phi=1\). Biologically, an interior equilibrium corresponds to a school environment in which susceptible students, casual users, problematic or SUD-level users, and resistant students all persist simultaneously. The condition
\[
\beta_{AC}<\mu+\delta
\]
has a direct interpretation: escalation from casual use into the problematic-use class is too weak to overcome school exit and disengagement, so sustained problematic use cannot persist. Here \(\beta_{AC}\) measures escalation pressure, while \(\mu+\delta\) combines ordinary turnover with problematic-use-related school disengagement. In contrast, when \(\beta_{AC}\) is sufficiently large relative to \(\mu+\delta\), and when the invasion condition \(\mathcal{R}_z^{(c)}>1\) holds, the model predicts that the problematic-use class can be maintained and that an interior equilibrium may exist. The formula
\[
y^*=\frac{\mu+\delta+\beta_{RA}w^*}{\beta_{AC}}
\]
also shows that the long-term level of casual use needed to sustain problematic use depends on the balance between escalation \(\beta_{AC}\), removal \(\mu+\delta\), and protective de-escalation \(\beta_{RA}w^*\). Biologically, this is consistent with evidence that repeated exposure to higher-risk peers promotes escalation, whereas protective peer environments and school-connected supports help interrupt progression to more severe use \cite{Watts2024,Stritzel2021,Hawkins1992,Schuler2019,Woodward2023}. From a policy perspective, this theorem suggests that schools should not only prevent initiation, but also weaken escalation pathways among already-using students. In model terms, lowering \(\beta_{AC}\), increasing \(\beta_{RA}\), reducing \(\delta\), and maintaining high re-engagement capacity are all mechanisms that push the system away from persistent coexistence of problematic use.

\medskip
\noindent\textbf{A numerical example illustrating multiple interior equilibria for the full system when \(\phi=1\).}
Consider \eqref{eq:scaledphi10} with
\[
\beta_{CS}=0.6589,\qquad
\beta_{AS}=4.1563,\qquad
\beta_{RC}=1.6324,\qquad
\beta_{AC}=3.9287,
\]
\[
\beta_{RA}=0.5436,\qquad
\beta_{RS}=0.5455,\qquad
\mu=0.3565,\qquad
\delta=0.1058,\qquad
\rho=0.1707.
\]
For this parameter set, the scalar equation \eqref{eq:w_scalar_phi1_clean} admits two admissible roots, yielding two interior equilibria:
\[
E_1\approx(0.1588,\,0.1489,\,0.4668,\,0.2255),
\qquad
E_2\approx(0.3458,\,0.1678,\,0.1240,\,0.3624).
\]
Numerically, \(E_1\) is locally asymptotically stable, whereas \(E_2\) is a saddle. In addition, the substance-free equilibrium
\[
E_f\approx(0.4505,\,0,\,0,\,0.5495)
\]
is also locally asymptotically stable. Thus the system again exhibits bistability between a stable substance-free equilibrium and a stable interior equilibrium, with the second interior equilibrium acting as a saddle separator.

For this parameter set, however,
\[
\Delta<0,
\]
so the SUD-free subsystem has no interior equilibrium and \(w_c^\pm\) are not defined. Hence this example lies outside the scope of Theorem~\ref{thm:yzw_persistence} and provides numerical evidence that the full system can admit multiple interior equilibria together with bistability.

\begin{corollary}[Sufficient conditions for extinction of problematic and casual use]
\label{cor:zy_to_zero}
Consider the reduced scaled SCAR system \eqref{eq:scaled_reduced} with \(0\le \phi\le 1\).

\begin{enumerate}
    \item If
    \begin{equation}\label{eq:cor_z_to_zero_cond_gen}
    \beta_{AC}<\mu+\delta,
    \end{equation}
    then
    \[
    \lim_{t\to\infty} z(t)=0.
    \]

    \item If, in addition,
    \begin{equation}\label{eq:cor_y_to_zero_cond_gen}
    w_{\mathrm{thr}}<w_c^-,
    \end{equation}
    so that the substance-free equilibrium \(E_f=(0,w_f)\) of the SUD-free subsystem is globally asymptotically stable by Theorem~\ref{thm:yw}(iv)(c), then
    \[
    \lim_{t\to\infty} y(t)=0,
    \qquad
    \lim_{t\to\infty} w(t)=w_f.
    \]
\end{enumerate}

Consequently, under \eqref{eq:cor_z_to_zero_cond_gen} and \eqref{eq:cor_y_to_zero_cond_gen},
\[
\lim_{t\to\infty}(y(t),z(t),w(t))=(0,0,w_f).
\]
That is, the reduced system converges to the substance-free equilibrium \(E_f=(0,0,w_f)\).
\end{corollary}

\noindent\textbf{Interpretation and policy implications.}
Corollary~\ref{cor:zy_to_zero} separates two levels of elimination. The condition
\[
\beta_{AC}<\mu+\delta
\]
is sufficient to eliminate problematic or SUD-level use, since escalation from casual use is too weak to overcome turnover and removal from the problematic-use class. However, eliminating problematic use alone does not guarantee elimination of substance use altogether. To ensure that casual use also disappears, the school environment must additionally satisfy conditions under which the substance-free equilibrium of the SUD-free subsystem is globally attracting, for example \(w_{\mathrm{thr}}<w_c^-\). Biologically, this means that once severe use has been removed, the remaining casual-use dynamics must still be unable to settle into a sustained positive level. From a policy perspective, the corollary shows that reducing severe use and eliminating substance use altogether are distinct goals: the former requires weakening escalation, whereas the latter also requires a sufficiently strong protective environment to drive the system all the way to the substance-free state.

\begin{remark}
The extinction criteria in Corollary~\ref{cor:zy_to_zero} are structurally the same for \(\phi=1\) and \(\phi<1\). The difference between the two regimes lies in the interpretation of the limiting state. When \(\phi=1\), the total population is conserved, so convergence to \(E_f\) corresponds to a true substance-free equilibrium in the original variables. When \(\phi<1\), the total population satisfies
\[
N'(t)=-(1-\phi)\delta A(t)\le0,
\]
so the system approaches a substance-free composition with a possibly reduced limiting population size due to permanent losses associated with problematic use.
\end{remark}

%%%%%%%%%%%%%%%%%%%%%%%%%%%%%%%%%%%%%%%%%%%%%%%%%%%
\section{Bifurcation diagrams, dynamical outcomes, and biological implications}

The bifurcation analysis complements the theoretical results by showing how one-parameter changes reorganize the equilibrium structure of the scaled SCAR system in the regime \(\phi=1\). Using numerical continuation, we examined representative one-parameter bifurcation diagrams for the problematic-use proportion. Across the parameter ranges explored, the dominant qualitative structure is bistability between a locally asymptotically stable substance-free equilibrium and a locally asymptotically stable interior equilibrium, separated by an interior saddle. Thus, the bifurcation diagrams show that long-term school-level outcomes may depend on both parameter values and initial conditions.

Several recurring patterns appear throughout the bifurcation diagrams. In some parameter regimes, the substance-free equilibrium is the unique attractor, indicating that the school environment cannot sustain long-term substance use. In many others, the system is bistable, with a stable substance-free equilibrium and a stable interior high-use equilibrium coexisting. In this case, the interior saddle acts as a separatrix between attraction to a substance-free state and attraction to a high-use interior state. For several protective or recovery-promoting parameters, the stable interior branch and the saddle branch collide and disappear through a saddle-node bifurcation, leaving only the substance-free equilibrium. In a few cases, especially for $\beta_{CS}$, the branch structure is richer and includes interactions among substance-free, SUD-free, and interior equilibria, indicating that some parameters reorganize the full equilibrium geometry rather than merely shifting equilibrium prevalence.

Figures~\ref{fig:bifurcation_main_2x2} and~\ref{fig:bifurcation_delta_rho_mu} summarize representative bifurcation structures for the baseline parameter set
\[
\mu = 0.3565,\quad
\rho = 0.1707,\quad
\delta = 0.1058,\quad
\phi = 1,
\]
\[
\beta_{AS} = 4.1563,\quad
\beta_{RS} = 0.5455,\quad
\beta_{RC} = 1.6324,\quad
\beta_{AC} = 3.9287,\quad
\beta_{RA} = 0.5436,\quad
\beta_{CS} = 0.6589.
\]
In all panels, blue branches denote locally asymptotically stable equilibria, green branches denote saddle equilibria, and red branches denote unstable equilibria.

Figure~\ref{fig:bifurcation_main_2x2} displays four representative bifurcation diagrams corresponding to escalation, problematic-peer recruitment, casual-user recruitment, and protective peer influence. Panel~\subref{fig:bif_betaAC_A} shows a typical escalation-driven threshold: as $\beta_{AC}$ increases, a stable interior equilibrium and a saddle are created through a saddle-node bifurcation, after which the system becomes bistable. Biologically, this means that if progression from casual use to problematic use becomes sufficiently strong, then a stable interior high-use equilibrium can arise even when the substance-free state remains locally stable. Thus, $\beta_{AC}$ identifies escalation from casual use to problematic use as a key mechanism governing the emergence of high-use school-level dynamics.

Panel~\subref{fig:bif_betaAS_A} highlights the role of $\beta_{AS}$, the influence of problematic-use students on susceptible students. This parameter is especially important in the school setting because it directly captures recruitment pressure from the problematic-use class. Its bifurcation structure shows that increasing $\beta_{AS}$ enlarges the parameter region supporting stable interior high-use equilibria. Biologically, stronger exposure of susceptible students to entrenched problematic users increases the likelihood of movement toward a high-use state. From a policy perspective, this suggests that school measures aimed at reducing the influence of the $A$ class may be critical. Such measures include rapid identification of high-risk clusters, targeted intervention for students already in the problematic-use class, increased supervision, and efforts to reduce opportunities for sustained exposure of susceptible students to severe-use peer networks.

Panel~\subref{fig:bif_betaCS_A} shows that $\beta_{CS}$, the influence of casual users on susceptible students, also plays a central role in the global dynamics. This parameter is particularly important because the casual-use class is not only a source of recruitment into use, but also the upstream pool from which students may later progress into the problematic-use class through the transition governed by $\beta_{AC}$. Increasing $\beta_{CS}$ enlarges the region supporting stable interior equilibria and generates some of the richest branch structure in the model. Thus, the bifurcation diagrams indicate that schools should not focus only on the problematic-use class. Restricting the prevalence and social influence of the casual-use class is also essential for protecting susceptible students and for cutting off the supply of future problematic users. In practice, this supports policies aimed at reducing visible casual use, limiting normalization of occasional use, and intervening early before experimentation becomes widespread.

Panel~\subref{fig:bif_betaRS_A} provides a representative protective bifurcation. As $\beta_{RS}$ increases, the stable interior branch shrinks and eventually disappears through collision with the interior saddle, leaving only the stable substance-free equilibrium. At the same time, the resistant component of the low-use equilibrium increases. Hence, protective peer influence does not merely lower equilibrium prevalence; it restructures the low-use state by enlarging the resistant subpopulation. The bifurcation diagrams for $\beta_{RA}$ and $\rho$ show qualitatively similar protective geometry, although their biological interpretations differ. Together, these parameters represent recovery support, peer protection, and resistant entry into the school environment. Their shared effect is to contract or eliminate the region supporting stable interior equilibria and enlarge the basin of attraction of the low-use state.

Figure~\ref{fig:bifurcation_delta_rho_mu} shows three additional representative bifurcation diagrams for $\delta$, $\rho$, and $\mu$. These parameters influence the system through demographic and school-structure mechanisms rather than direct peer-contact pathways. Panel~\subref{fig:bif_delta_A} shows that increasing $\delta$ generally weakens the stable interior problematic-use branch and can eliminate the interior equilibrium entirely. Mathematically, this appears stabilizing, but biologically its interpretation must be made with caution: a larger $\delta$ may reflect successful removal from substance use through intervention, but it may also reflect school disengagement or dropout. Hence, reducing the high-use branch through larger $\delta$ is not necessarily a desirable outcome unless accompanied by rehabilitation and successful return. In this sense, the $\delta$-bifurcation highlights one of the practical strengths of the model: it distinguishes reductions in problematic use caused by prevention and recovery from reductions caused by educational disengagement.

Panel~\subref{fig:bif_rho_A} shows that increasing $\rho$, the fraction of students entering the school already strongly resistant to substance use, shrinks the region supporting interior equilibria and enlarges the basin of attraction of the substance-free equilibrium. Biologically, this indicates that prevention prior to or at school entry can have a strong system-wide effect. In particular, upstream prevention, family support, community norms, and early school-based orientation programs may reduce the likelihood that the school system moves toward a high-use regime.

Panel~\subref{fig:bif_mu_A} shows that the turnover parameter $\mu$ has a more subtle effect. Because $\mu$ governs both entry and exit from the school population, it acts as a demographic time-scale parameter rather than a direct behavioral parameter. Its bifurcation diagram indicates that changes in student turnover can shift the equilibrium structure and alter bistability, often in a non-monotone way. Biologically, this means that school mobility, graduation flow, and replacement of student cohorts can modify long-term substance-use dynamics even when peer-influence parameters remain unchanged. Thus, $\mu$ is best interpreted as a contextual parameter that shapes the effective exposure time of students to risk and protection within the school environment.

Taken together, the bifurcation diagrams suggest three broad classes of parameter effects. The first consists of the principal risk-promoting parameters $\beta_{AC}$, $\beta_{AS}$, and $\beta_{CS}$. Increasing these parameters enlarges the region where stable interior equilibria exist, but through distinct mechanisms. The parameter $\beta_{AC}$ acts through escalation from casual use to problematic use, $\beta_{AS}$ acts through direct recruitment pressure from problematic users onto susceptible students, and $\beta_{CS}$ acts through the spread of use by the casual-use class. Since casual users can both recruit susceptible students and later progress into the problematic-use class, the bifurcation diagrams identify the casual-use population as an especially important policy target.

The second class consists of protective or recovery-promoting parameters such as $\beta_{RA}$, $\beta_{RC}$, $\beta_{RS}$, and $\rho$. Increasing these parameters shrinks or removes the stable interior equilibrium, although the mechanisms differ. The parameters $\beta_{RA}$ and $\beta_{RC}$ promote de-escalation and recovery, whereas $\beta_{RS}$ and $\rho$ strengthen resistance by building a larger protected subpopulation. In particular, $\rho$ represents an upstream protective mechanism, indicating the importance of prevention before or at school entry.

The third class consists of contextual or structurally mediated parameters such as $\delta$ and $\mu$. Increasing $\delta$ weakens the interior equilibrium mathematically, but its biological meaning depends on whether exit from the problematic-use class represents supported recovery or school dropout. The turnover parameter $\mu$ has non-monotone effects, indicating that it functions primarily as a demographic or time-scale parameter rather than as a direct intervention lever. Thus, while $\delta$ and $\mu$ clearly shape the bifurcation structure, their policy interpretation is more nuanced than that of the direct peer-influence or protective parameters.

The usefulness of the model lies in its ability to connect school-relevant mechanisms to qualitative shifts in long-term outcomes. The bifurcation diagrams show that adolescent substance use is not governed solely by a single invasion threshold. Rather, the system frequently exhibits multistability, so prevention and intervention must account for both prevalence reduction and basin management. In particular, the model identifies several practically meaningful control pathways: reducing escalation from casual to problematic use ($\beta_{AC}$), reducing recruitment pressure from problematic users ($\beta_{AS}$), reducing the prevalence and influence of the casual-use class ($\beta_{CS}$), and strengthening protective peer influence and resistance-building processes ($\beta_{RS}$, $\rho$). In this sense, the model is useful not only for explaining threshold-dependent high-use outcomes, but also for identifying which school-level mechanisms are most likely to move the system away from high-use regimes and toward durable low-use outcomes.

From a policy perspective, the bifurcation structure supports a layered intervention strategy. Schools should simultaneously reduce peer-driven initiation and escalation, especially by limiting recruitment pressure from both casual and problematic users, and strengthen protective culture through mentoring, supervision, prevention programs, and early intervention. The diagrams in panels~\subref{fig:bif_betaAS_A} and~\subref{fig:bif_betaCS_A} indicate that controlling the influence of both problematic and casual users is critical for preventing the emergence of a high-use state. More broadly, the model suggests that effective school policy must aim not only to reduce average substance use, but also to keep the system away from threshold regions where modest increases in harmful peer influence can trigger transition from a low-use state to a high-use state.

Table~\ref{tab:bif_summary} summarizes the dominant bifurcation roles of each parameter together with their dynamical, biological, and policy interpretations.

{\footnotesize
\begin{table}[!htbp]
\centering
\caption{Summary of parameter effects inferred from the bifurcation diagrams of the scaled SCAR system ($\phi=1$).}
\label{tab:bif_summary}
\renewcommand{\arraystretch}{1.2}
\begin{tabular}{|c|p{3.2cm}|p{3.8cm}|p{4.5cm}|}
\hline
\textbf{Parameter} & \textbf{Typical bifurcation role} & \textbf{Dynamical effect} & \textbf{Biological / policy interpretation} \\
\hline
$\beta_{AC}$ & Risk-promoting threshold parameter & Creates stable interior equilibria and bistability & Escalation from casual to problematic use; prevent progression early \\
\hline
$\beta_{AS}$ & Risk-promoting threshold parameter & Expands the region supporting stable interior equilibria & Strong influence of problematic-use students on susceptible students; reduce exposure to entrenched high-risk peer groups \\
\hline
$\beta_{CS}$ & Structural organizer of multistability & Generates rich branch interactions and multiple thresholds & Casual users recruit susceptibles and supply the pool of future problematic users; reduce normalization and prevalence of casual use \\
\hline
$\beta_{RA}$ & Recovery-promoting parameter & Shrinks or removes stable interior equilibria & Promote de-escalation from problematic use; strengthen recovery support \\
\hline
$\beta_{RC}$ & Early recovery parameter & Shrinks and may eliminate stable interior equilibria & Reverse casual use before escalation; expand early intervention \\
\hline
$\beta_{RS}$ & Protective resilience-building parameter & Suppresses the interior branch and enlarges the resistant low-use state & Build protective peer culture and anti-use norms \\
\hline
$\rho$ & Upstream protective-entry parameter & Reduces the region supporting interior equilibria and enlarges the substance-free basin & Prevention before or at school entry; strengthen resistant attitudes at baseline \\
\hline
$\delta$ & Structurally suppressive parameter & Weakens or eliminates the stable interior high-use branch & May reduce problematic use mathematically, but can reflect dropout rather than recovery; interpret with caution \\
\hline
$\mu$ & Contextual demographic parameter & Non-monotone, regime-dependent effects & School turnover alters exposure duration, mixing, and cohort replacement \\
\hline
\end{tabular}
\end{table}
}

\begin{figure}[!htbp]
    \centering
    \subfloat[$\beta_{AC}$: escalation from casual use to problematic use.\label{fig:bif_betaAC_A}]{
        \includegraphics[width=0.4\textwidth]{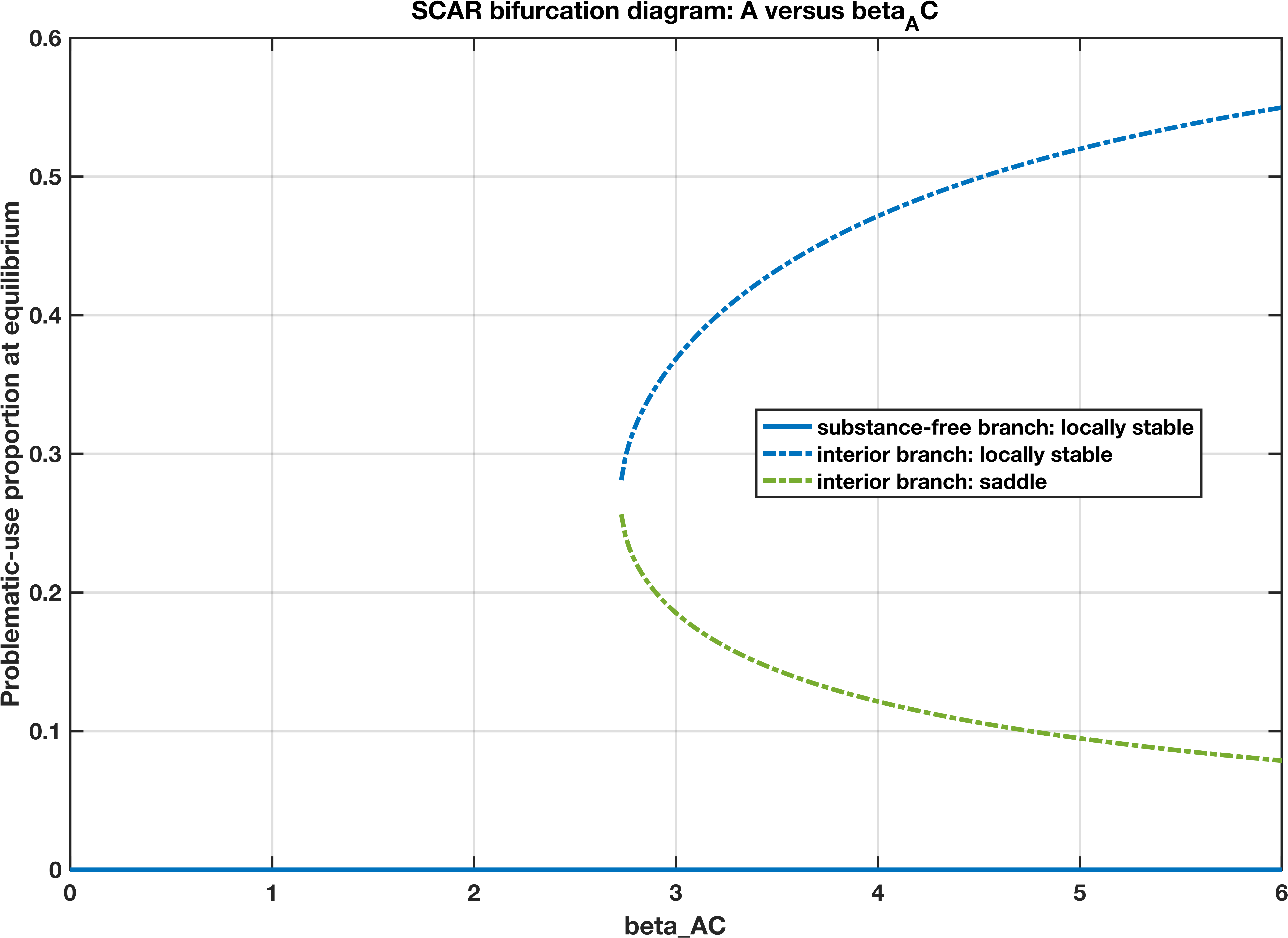}
    }
    \hfill
    \subfloat[$\beta_{AS}$: influence of problematic-use students on susceptible students.\label{fig:bif_betaAS_A}]{
        \includegraphics[width=0.4\textwidth]{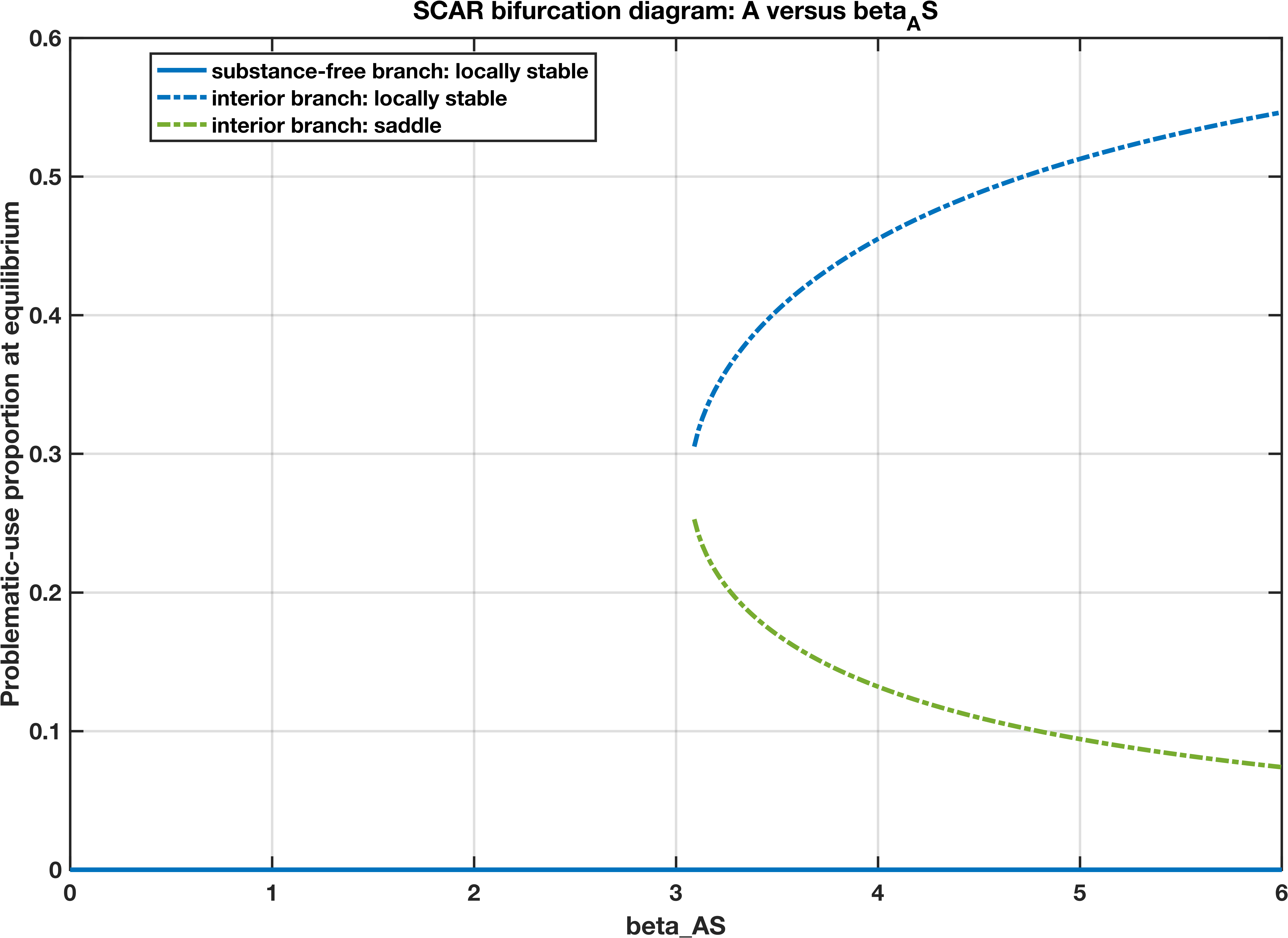}
    }

    \vspace{0.5em}

    \subfloat[$\beta_{CS}$: influence of casual users on susceptible students.\label{fig:bif_betaCS_A}]{
        \includegraphics[width=0.4\textwidth]{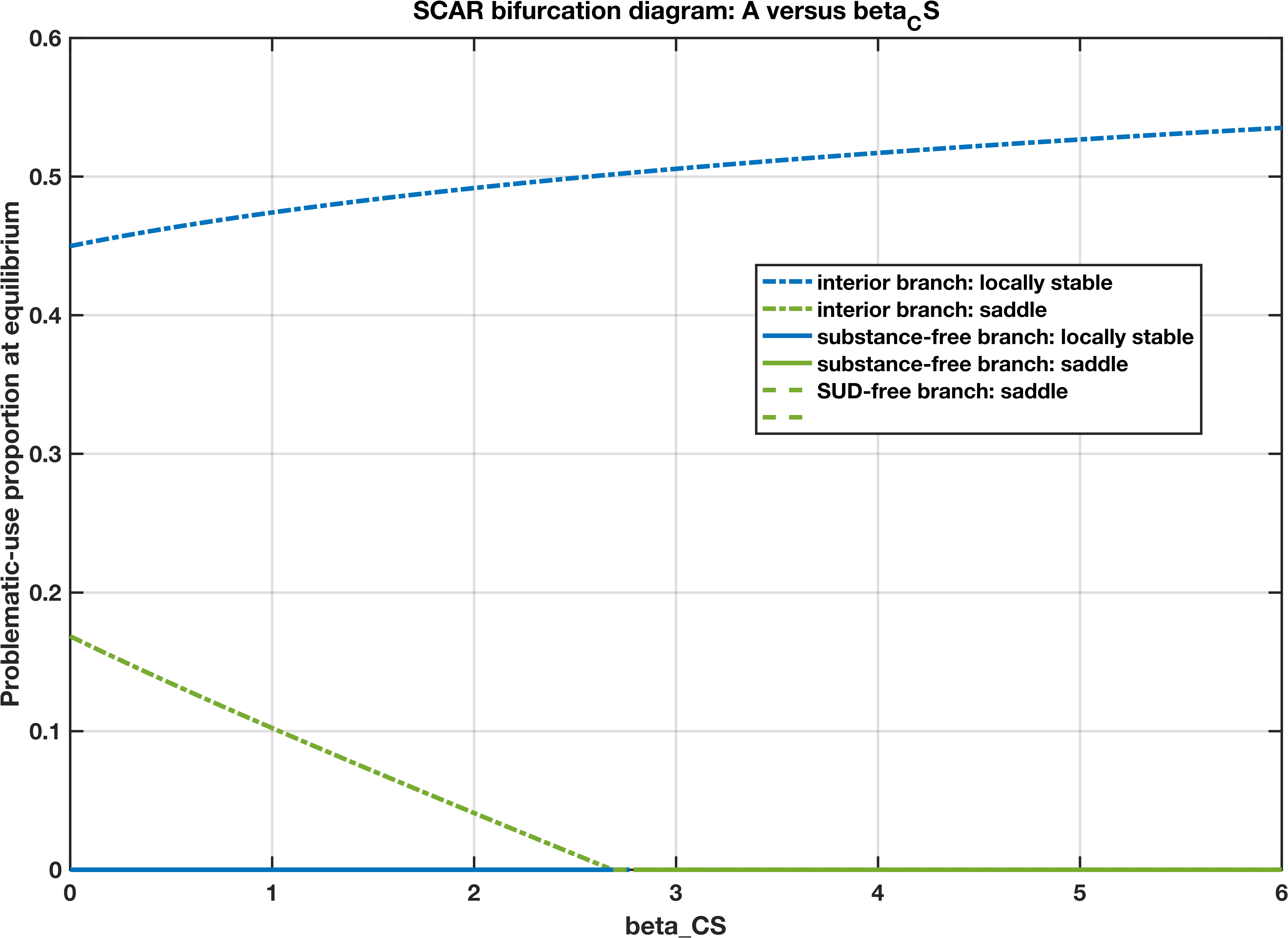}
    }
    \hfill
    \subfloat[$\beta_{RS}$: protective influence from resistant students.\label{fig:bif_betaRS_A}]{
        \includegraphics[width=0.4\textwidth]{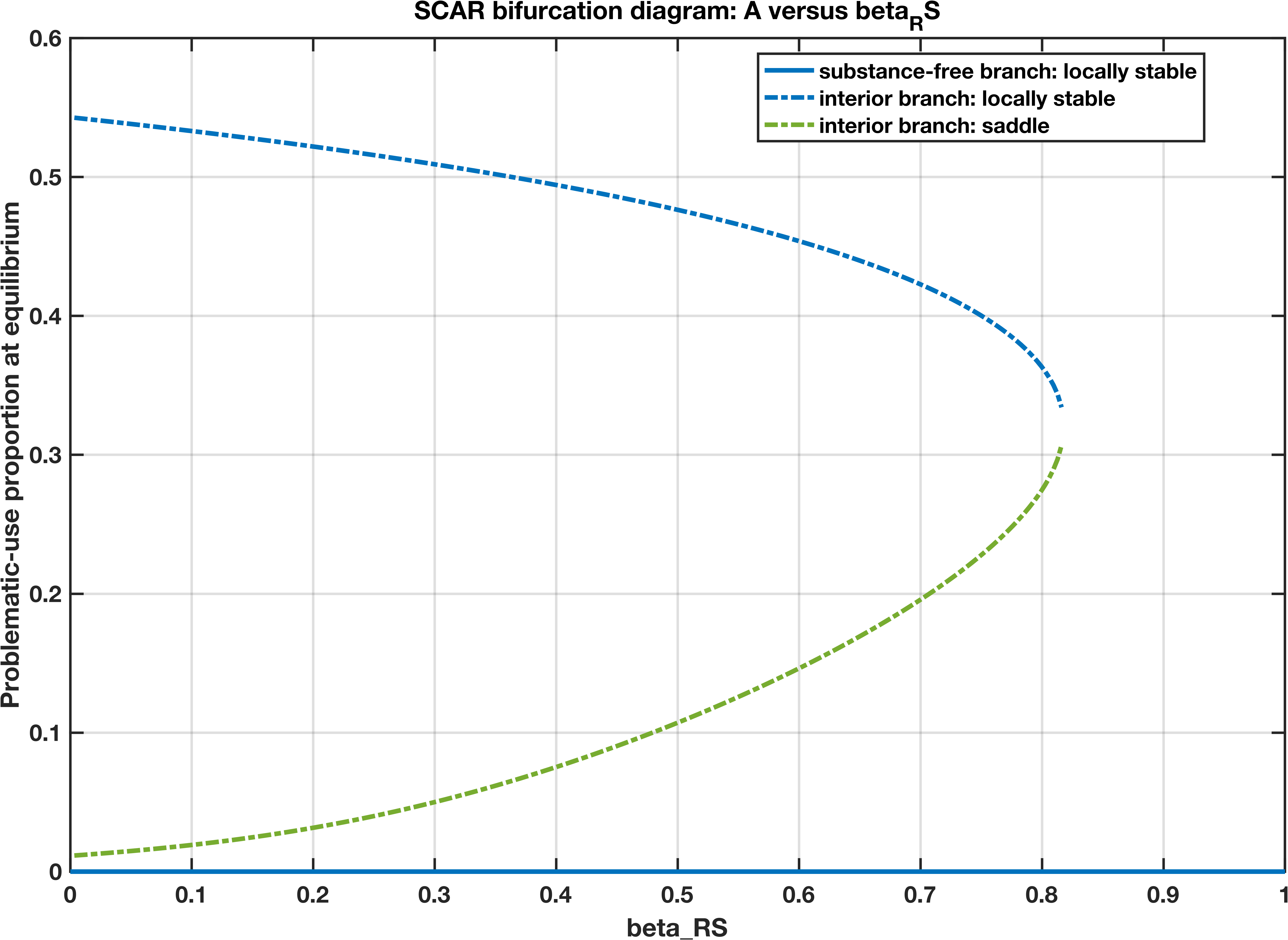}
    }

    \caption{Representative one-parameter bifurcation diagrams for the problematic-use proportion $A$ in the scaled SCAR system with $\phi=1$. Blue branches denote locally asymptotically stable equilibria, green branches denote saddle equilibria, and red branches denote unstable equilibria. The panels illustrate four distinct mechanisms shaping the long-term dynamics: escalation from casual use to problematic use ($\beta_{AC}$), recruitment pressure from problematic users ($\beta_{AS}$), recruitment pressure from casual users ($\beta_{CS}$), and protective peer influence ($\beta_{RS}$).}
    \label{fig:bifurcation_main_2x2}
\end{figure}

\begin{figure}[!htbp]
    \centering
    \subfloat[$\delta$: school disengagement / removal from the problematic-use class.\label{fig:bif_delta_A}]{
        \includegraphics[width=0.3\textwidth]{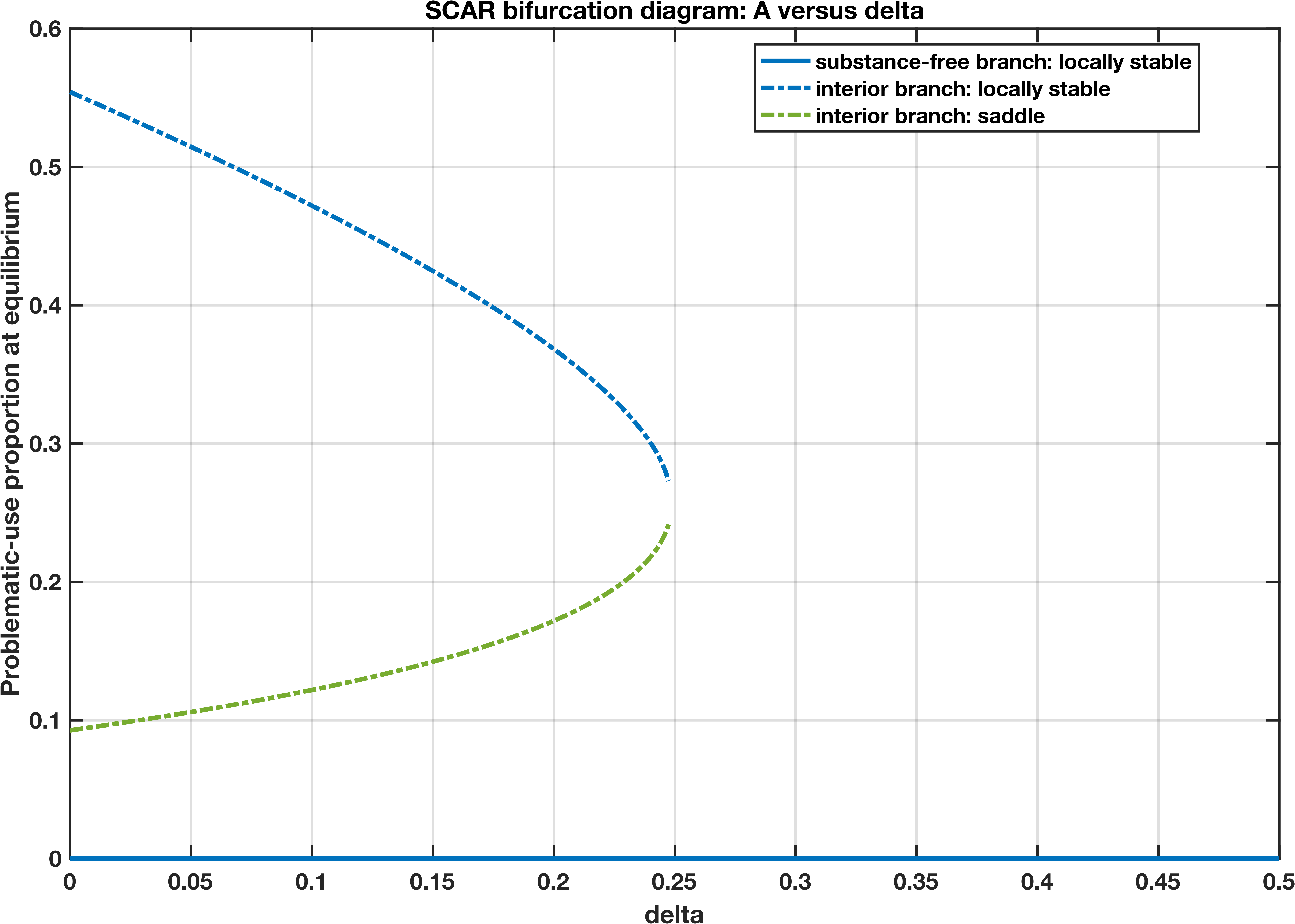}
    }
    \hfill
    \subfloat[$\rho$: fraction entering with strong resistance to substance use.\label{fig:bif_rho_A}]{
        \includegraphics[width=0.3\textwidth]{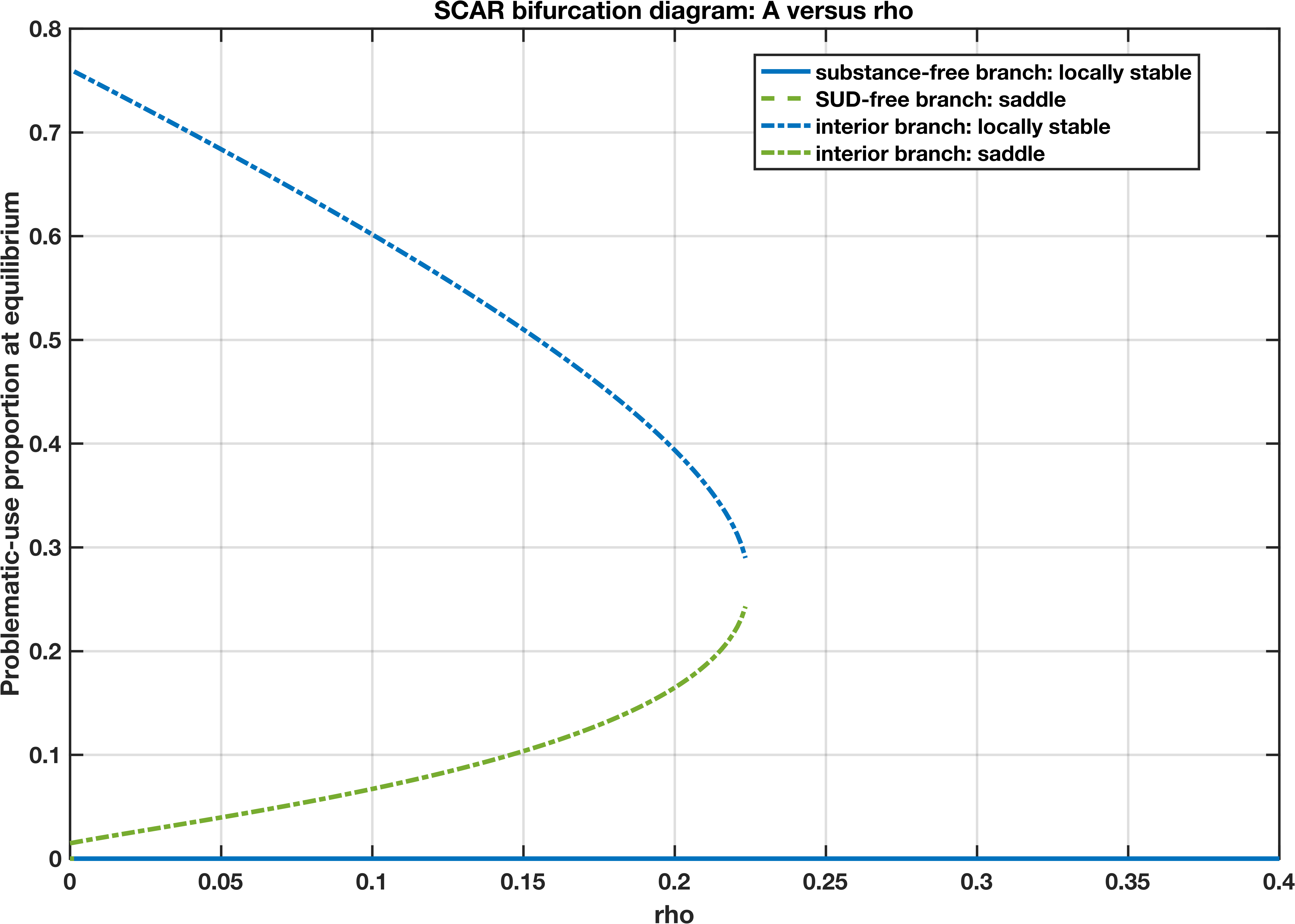}
    }
    \hfill
    \subfloat[$\mu$: school turnover rate.\label{fig:bif_mu_A}]{
        \includegraphics[width=0.3\textwidth]{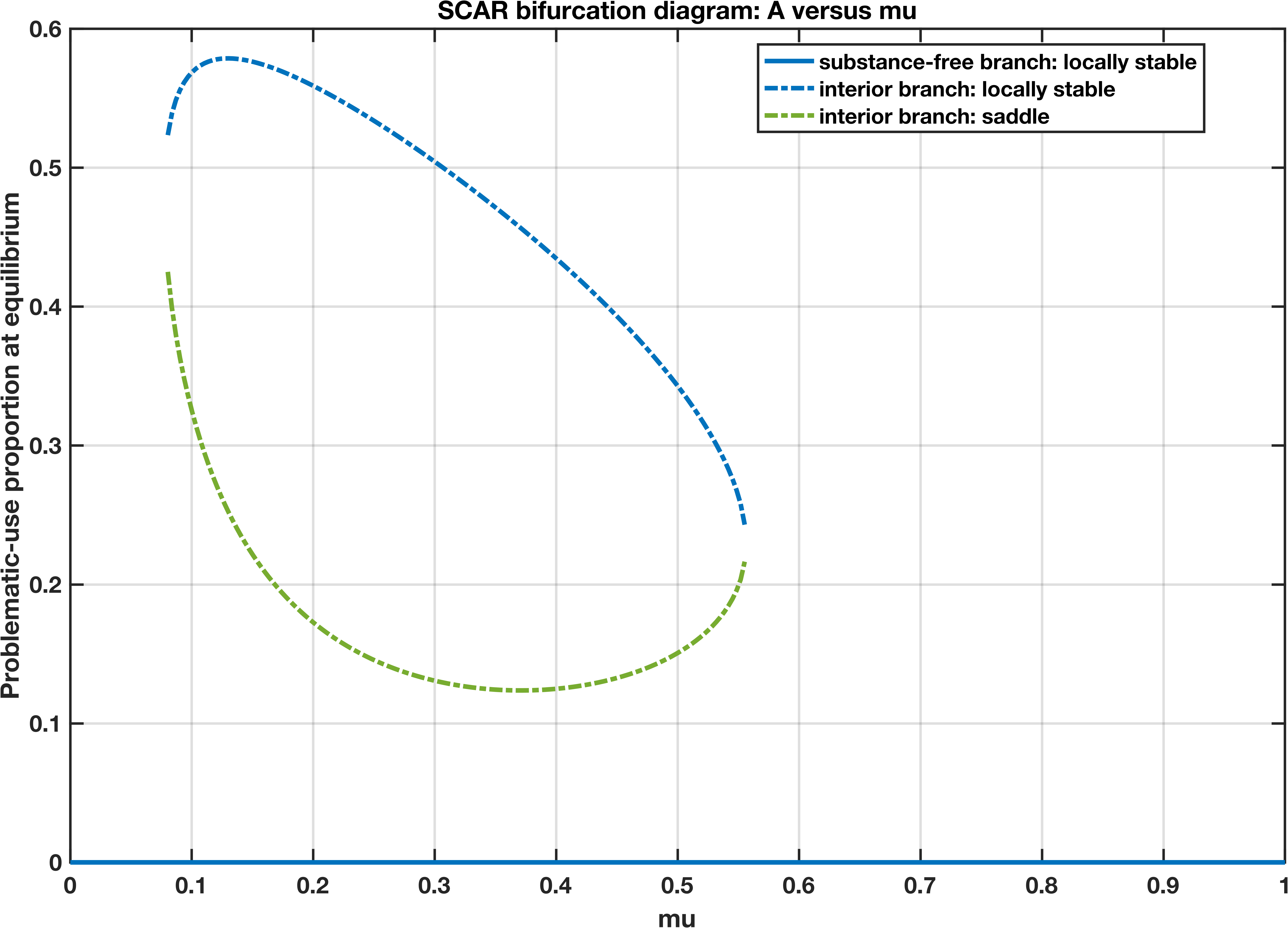}
    }

    \caption{Additional representative bifurcation diagrams for the problematic-use proportion $A$ in the scaled SCAR system with $\phi=1$. Blue branches denote locally asymptotically stable equilibria, green branches denote saddle equilibria, and red branches denote unstable equilibria. The panels illustrate the effects of the disengagement/removal rate $\delta$, the resistant-entry fraction $\rho$, and the school turnover rate $\mu$.}
    \label{fig:bifurcation_delta_rho_mu}
\end{figure}

%%%%%%%%%%%%%%%%%%%%%%%%%%%%%%%%%%%%%%%%%%%%%%%%%%%

\section{Discussion}

This paper develops a school-based SCAR model for adolescent substance use that links peer-driven behavioral transitions with school disengagement and re-entry. The main contribution is not empirical estimation, but a qualitative dynamical characterization of how initiation, escalation, protection, and return interact to determine long-term school-level outcomes.

At the theoretical level, three conclusions are central. First, the model is positively invariant and bounded, so all solutions remain well defined and remain in the biologically relevant region (Theorem~\ref{posbound}). Second, the return parameter \(\phi\) changes the meaning of the long-term dynamics: when \(\phi=1\), interior equilibria may represent genuine coexistence states, whereas when \(\phi<1\), positive equilibria of the scaled system do not correspond to true endemic equilibria in the original variables because problematic use induces net school-population loss. Third, the model separates initiation from escalation. The threshold \(\mathcal R_y\) governs invasion of casual use into the substance-free equilibrium, while \(\mathcal R_z^{(c)}\) governs invasion of problematic use into the SUD-free equilibrium (Theorem~\ref{thm:yzw_stability}). Thus, first use and progression to problematic use are dynamically distinct processes.

The subsystem results sharpen this interpretation. The SUD-free subsystem may converge to the substance-free equilibrium, to a stable casual-use equilibrium, or exhibit bistability between those two states (Theorem~\ref{thm:yw}). In addition, Corollary~\ref{cor:zy_to_zero} distinguishes elimination of problematic use from elimination of substance use altogether: \(\beta_{AC}<\mu+\delta\) is sufficient for \(z(t)\to0\), but extinction of casual use additionally requires that the substance-free equilibrium of the SUD-free subsystem be globally attracting. Thus, suppressing escalation is necessary but not sufficient by itself to eliminate school-level substance use.

The full reduced system adds a further layer through multistability. The permanence result for the reduced \(yzw\)-system (Theorem~\ref{thm:yzw_persistence}) and the interior-equilibrium characterization for \(\phi=1\) (Theorem~\ref{thm:phi1_interior_yzw}) show that the full model can support stable interior states with \(A^*>0\). The numerical examples in Section~3 further indicate that bistability between the substance-free equilibrium and a stable interior equilibrium can occur even outside the sufficient hypotheses of Theorem~\ref{thm:yzw_persistence}. Thus, the eventual school outcome may depend on both structural conditions and initial behavioral composition.

The bifurcation analysis in Figures~\ref{fig:bifurcation_main_2x2} and~\ref{fig:bifurcation_delta_rho_mu} clarifies which mechanisms are most responsible for these regime shifts. Panels~\subref{fig:bif_betaAC_A}--\subref{fig:bif_betaCS_A} show that the main risk-promoting parameters \(\beta_{AC}\), \(\beta_{AS}\), and \(\beta_{CS}\) enlarge the region supporting stable interior high-use equilibria, but through distinct pathways: escalation, recruitment from problematic users, and recruitment from casual users. Among these, \(\beta_{CS}\) is especially important because the casual-use class both recruits susceptible students and supplies the pool from which future problematic users emerge. By contrast, the protective parameters \(\beta_{RA}\), \(\beta_{RC}\), \(\beta_{RS}\), and \(\rho\) contract or eliminate the stable interior branch and enlarge the basin of attraction of low-use outcomes, while \(\delta\) and \(\mu\) act more indirectly through school structure and turnover. These qualitative roles are summarized in Table~\ref{tab:bif_summary}.

Taken together, the theoretical and bifurcation results imply that school-level substance use is governed by both threshold effects and basin structure. A school may therefore move from a low-use regime to a high-use regime either because a parameter crosses a bifurcation threshold or because the state is perturbed across a separatrix in a bistable regime. This provides a dynamical interpretation of early-use clusters, weakening of protective peer culture, and delayed intervention: such events may change not only prevalence levels, but also the eventual attractor of the system.

These results support a layered policy interpretation. Reducing \(\beta_{CS}\) and \(\beta_{AS}\) corresponds to limiting recruitment pressure from casual and problematic users through supervision, peer-norm correction, targeted prevention, and rapid response to emerging clusters of use. Reducing \(\beta_{AC}\) corresponds to weakening escalation among already-using students through early screening, counseling, and treatment-linked intervention. Increasing \(\beta_{RS}\), \(\beta_{RC}\), \(\beta_{RA}\), and \(\rho\) corresponds to strengthening protective peer influence, recovery-supportive environments, and resistant entry into the school system through mentoring, family engagement, prevention programming, and re-entry support. The role of \(\delta\) must be interpreted especially carefully: although increasing \(\delta\) may weaken the high-use branch mathematically, this is not necessarily favorable if it represents dropout rather than successful treatment and return.

The model nevertheless has important limitations. It assumes homogeneous mixing and frequency-dependent contact, whereas real school populations are structured by friendship networks, grade levels, and activity-based interactions. It also compresses re-entry dynamics into a single parameter \(\phi\), although in practice return may depend on treatment access, delays, family support, and school policy. In addition, the distinction between casual and problematic use is a useful but simplified representation of a broader behavioral continuum. These assumptions are appropriate for a baseline dynamical framework, but they limit literal quantitative interpretation.

Future work should extend the model toward greater structural realism through network-based contact patterns, grade stratification, delayed re-entry, or time-dependent prevention and intervention terms. Calibration against longitudinal school or district data would also help connect the present qualitative analysis to measurable school-level outcomes and make the framework more directly useful for data-informed prediction and intervention design.

Overall, the SCAR framework shows that adolescent substance use in schools is governed not only by invasion thresholds, but also by escalation dynamics, multistability, and re-entry structure. Effective intervention therefore requires not only reducing average risk, but also keeping the system away from parameter regimes and state-space regions that support sustained higher-use outcomes.
%%%%%%%%%%%%%%%%%%%%%%%%%%%%%%%%%%%%%%%%%%%%%%%%%%%
\section{Detailed Proofs of Theorems}

\subsection*{Proof of Theorem \ref{posbound}}
\begin{proof}
\textbf{Step 1. Positivity of \(A(t)\).}
For \(A(t)\) we have
\[
\frac{dA}{dt}
= \frac{\beta_{AC}AC}{N} - (\mu+\delta)A - \frac{\beta_{RA}AR}{N}
= A\Big(\frac{\beta_{AC}C}{N} -\mu -\delta -\frac{\beta_{RA}R}{N} \Big).
\]
This is a linear ODE for \(A\) of the form \(A' = A q(t)\), where
\[
q(t)=\frac{\beta_{AC}C(t)}{N(t)}-\mu-\delta-\frac{\beta_{RA}R(t)}{N(t)}
\]
is continuous as long as \(N(t)>0\). Hence
\[
A(t) = A(0)\exp\!\left( \int_0^t q(s)\, ds\right) \ge 0.
\]
If \(A(0)=0\), then \(A(t)\equiv 0\). If \(A(0)>0\), then \(A(t)>0\) for all \(t>0\).

\medskip
\textbf{Step 2. Positivity of \(N(t)\).}
Summing equations~\eqref{eq:S}--\eqref{eq:R} yields
\[
\frac{dN}{dt} = -(1-\phi)\delta A(t)\le 0.
\]
Since \(A(t)\ge 0\), it follows that \(N(t)\) is nonincreasing and
\[
0 \le N(t) \le N(0).
\]
Because \(A(t)\) cannot become negative, the right-hand side of the equation for \(N\) is bounded below, and \(N(t)\) cannot reach zero in finite time provided \(N(0)>0\). Therefore \(N(t)>0\) for all \(t\ge 0\).

\medskip
\textbf{Step 3. Positivity of \(R(t)\).}
The \(R\)-equation is
\[
\frac{dR}{dt} = \mu \rho N + \frac{\beta_{RS}RS}{N} - \mu R.
\]
If \(R(t^*)=0\) at some time \(t^*\), then
\[
R'(t^*)=\mu\rho N(t^*)\ge 0.
\]
Therefore \(R(t)\) cannot cross below zero, and hence \(R(t)\ge 0\) for all \(t\ge 0\).

\medskip
\textbf{Step 4. Positivity of \(S(t)\) and \(C(t)\).}
Consider the equation for \(S\):
\[
\frac{dS}{dt}
= \mu(1-\rho)N
    - \frac{\beta_{CS}CS}{N}
    - \frac{\beta_{AS}AS}{N}
    - \mu S
    - \frac{\beta_{RS}RS}{N}
    + \frac{\beta_{RC}RC}{N}
    + \phi\delta A.
\]
Since the initial conditions of all components are nonnegative, without loss of generality, assume that there is a time \(t^*\) such that \(S(t^*)=0\) such that \(C(t^*)>0\), we obtain
\[
S'(t^*) = \mu(1-\rho)N(t^*) + \phi\delta A(t^*) + \frac{\beta_{RC}R(t^*)C(t^*)}{N(t^*)}\ge 0.
\]
Thus \(S(t)\) cannot cross below zero.

For \(C(t)\),
\[
\frac{dC}{dt}
= \frac{\beta_{CS}CS}{N}
  + \frac{\beta_{AS}AS}{N}
  - \mu C
  - \frac{\beta_{RC}RC}{N}
  - \frac{\beta_{AC}AC}{N}
  + \frac{\beta_{RA}AR}{N}.
\]
At any time \(t^*\) such that \(C(t^*)=0\), we obtain
\[
C'(t^*) = \frac{\beta_{AS}A(t^*)S(t^*)}{N(t^*)}
          + \frac{\beta_{RA}A(t^*)R(t^*)}{N(t^*)} \ge 0.
\]
Hence \(C(t)\) cannot cross below zero.

Similar proof can be applied in the case that $C(t)$ is the first one to have a $t^*$ such that $C(t^*)=0$ with $S(t^*)\geq 0$.

\medskip
\textbf{Step 5. Boundedness.}
Since \(S,C,A,R\ge 0\) and
\[
N(t)=S(t)+C(t)+A(t)+R(t),
\]
and since \(N(t)\le N(0)\) for all \(t\ge 0\), each compartment satisfies
\[
0\le S(t),C(t),A(t),R(t)\le N(t)\le N(0).
\]
Thus every solution remains in the compact region
\[
\Omega=\{(S,C,A,R)\in\mathbb{R}_{+}^{4}: S+C+A+R\le N(0)\}.
\]
Therefore \(\Omega\) is positively invariant and bounded.
\end{proof}

%%%%%%%%%%%%%%%%%%%%%%%%%%%%%%%%%%%%%%%%%%%%%%%%%%%%%%%%%%%%%%%%%%%%%%%%%%%%%%%%%%%%
\subsection*{Proof of Theorem \ref{thm:drugfree_wf}}
\begin{proof}
Define
\[
f(w)=\mu\rho-(\mu-\beta_{RS})w-\beta_{RS}w^2.
\]
Then \eqref{eq:scaledphi1yz0} is the scalar ODE
\[
w'=f(w),
\]
with \(f\) a polynomial, hence locally Lipschitz. Therefore solutions exist and are unique for all \(t\ge 0\).

\medskip
\noindent\textbf{Step 1. Positive invariance of \([0,1]\).}
We evaluate the vector field at the endpoints:
\[
f(0)=\mu\rho\ge 0,\qquad f(1)=\mu\rho-\mu\le 0.
\]
Thus, at \(w=0\) the vector field points into \([0,1]\) (or is tangent if \(\rho=0\)), and at \(w=1\) the vector field points into \([0,1]\) (or is tangent if \(\rho=1\)). By uniqueness of solutions, no trajectory starting in \([0,1]\) can cross the boundary, and therefore \([0,1]\) is positively invariant.

\medskip
\noindent\textbf{Step 2. Existence and uniqueness of the equilibrium in \([0,1]\).}
Equilibria satisfy \(f(w)=0\), that is,
\[
\beta_{RS}w^2+(\mu-\beta_{RS})w-\mu\rho=0.
\]
Solving this quadratic gives
\[
w_{\pm}=\frac{(\beta_{RS}-\mu)\pm\sqrt{(\beta_{RS}-\mu)^2+4\rho\mu\beta_{RS}}}{2\beta_{RS}}.
\]
The positive root is
\[
w_f=\frac{(\beta_{RS}-\mu)+\sqrt{(\beta_{RS}-\mu)^2+4\rho\mu\beta_{RS}}}{2\beta_{RS}}.
\]
For the other root, note that
\[
\sqrt{(\beta_{RS}-\mu)^2+4\rho\mu\beta_{RS}}\ge |\beta_{RS}-\mu|,
\]
with strict inequality if \(\rho>0\). Hence
\[
(\beta_{RS}-\mu)-\sqrt{(\beta_{RS}-\mu)^2+4\rho\mu\beta_{RS}}\le 0,
\]
so \(w_-\le 0\), with \(w_-<0\) if \(\rho>0\). Therefore there is at most one biologically relevant nonnegative equilibrium, namely \(w_f\).

Since \(f(0)\ge 0\) and \(f(1)\le 0\), the intermediate value theorem guarantees at least one root in \([0,1]\). Hence \(w_f\in[0,1]\) and is the unique equilibrium in \([0,1]\).

\medskip
\noindent\textbf{Step 3. Global asymptotic stability on \([0,1]\).}
Because \(\beta_{RS}>0\),
\[
f''(w)=-2\beta_{RS}<0,
\]
so \(f\) is strictly concave on \(\mathbb{R}\). Since \(f\) has exactly one zero in \([0,1]\), namely \(w=w_f\), strict concavity implies that
\[
f(w)>0\quad\text{for }0\le w<w_f,
\qquad
f(w)<0\quad\text{for }w_f<w\le 1.
\]
Let \(w(0)\in[0,1]\). If \(w(0)<w_f\), then \(w'(t)>0\) as long as \(w(t)<w_f\), so \(w(t)\) is increasing and bounded above by \(w_f\). Hence \(w(t)\) converges to a limit \(\ell\le w_f\). Passing to the limit in the equation \(w'=f(w)\) gives \(f(\ell)=0\), so \(\ell=w_f\).

If instead \(w(0)>w_f\), then \(w'(t)<0\) as long as \(w(t)>w_f\), so \(w(t)\) is decreasing and bounded below by \(w_f\). Hence \(w(t)\) converges to a limit \(\ell\ge w_f\), and again \(f(\ell)=0\) implies \(\ell=w_f\).

Therefore, for any initial condition \(w(0)\in[0,1]\),
\[
\lim_{t\to\infty}w(t)=w_f.
\]
This proves global asymptotic stability of \(w_f\) on \([0,1]\).
\end{proof}

%%%%%%%%%%%%%%%%%%%%%%%%%%%%%%%%%%%%%%%%%%%%%%%%%%%%%%%%%%%%%%%%%%%%%%%%%%%%%%%%%%%%
\subsection*{Proof of Theorem \ref{thm:yw}}
\begin{proof}
Let
\[
f(y,w)=\dot y,\qquad g(y,w)=\dot w.
\]
Then the Jacobian matrix of the scaled SUD-free subsystem \eqref{eq:scaledphi1z0} is
\[
J(y,w)=
\begin{pmatrix}
f_y & f_w\\[3pt]
g_y & g_w
\end{pmatrix}
=
\begin{pmatrix}
\beta_{CS}(1-2y-w)-\mu-\beta_{RC}w & -y(\beta_{CS}+\beta_{RC})\\[4pt]
-\beta_{RS}w & \beta_{RS}(1-y-2w)-\mu
\end{pmatrix}.
\]

\medskip
\noindent\textbf{(i) Boundary equilibrium \(E_f=(0,w_f)\).}
Setting \(y=0\) and \(w'=0\) yields the unique positive boundary equilibrium
\[
E_f=(0,w_f),
\qquad
w_f=\frac{(\beta_{RS}-\mu)+\sqrt{(\beta_{RS}-\mu)^2+4\rho\mu\beta_{RS}}}{2\beta_{RS}}.
\]
At \(E_f\), the Jacobian becomes
\[
J(0,w_f)=
\begin{pmatrix}
\beta_{CS}(1-w_f)-\mu-\beta_{RC}w_f & 0\\[4pt]
-\beta_{RS}w_f & \beta_{RS}(1-2w_f)-\mu
\end{pmatrix}.
\]
The eigenvalues are
\[
\lambda_1=\beta_{CS}(1-w_f)-\mu-\beta_{RC}w_f,\qquad
\lambda_2=\beta_{RS}(1-2w_f)-\mu.
\]
Since
\[
w_f>\frac{\beta_{RS}-\mu}{2\beta_{RS}},
\]
we obtain
\[
\lambda_2
=
\beta_{RS}(1-2w_f)-\mu
=
2\beta_{RS}\left(\frac{\beta_{RS}-\mu}{2\beta_{RS}}-w_f\right)<0.
\]
Define
\[
w_{\mathrm{thr}}=\frac{\beta_{CS}-\mu}{\beta_{CS}+\beta_{RC}}.
\]
Then
\[
\lambda_1<0
\iff
\beta_{CS}(1-w_f)-\mu-\beta_{RC}w_f<0
\iff
w_f>w_{\mathrm{thr}}.
\]
Therefore, the boundary equilibrium \(E_f\) is locally asymptotically stable if and only if \(w_f>w_{\mathrm{thr}}\).

\medskip
\noindent\textbf{(ii) Interior equilibria.}
At any interior equilibrium \((y^*,w^*)\), we must have \(y^*>0\), so \(y'=0\) implies
\[
\beta_{CS}(1-y^*-w^*)-\mu-\beta_{RC}w^*=0,
\]
hence
\[
1-y^*-w^*=\frac{\mu+\beta_{RC}w^*}{\beta_{CS}},
\qquad
y^*=1-w^*-\frac{\mu+\beta_{RC}w^*}{\beta_{CS}}.
\]
Substituting this into \(w'=0\) gives
\[
\mu\rho-w^*\left[\mu-\beta_{RS}\left(\frac{\mu+\beta_{RC}w^*}{\beta_{CS}}\right)\right]=0.
\]
After simplification,
\[
\frac{\beta_{RS}\beta_{RC}}{\beta_{CS}}(w^*)^2
-\mu\left(1-\frac{\beta_{RS}}{\beta_{CS}}\right)w^*
+\mu\rho=0.
\]
Thus interior equilibria are determined by the quadratic equation
\[
F(w):=\frac{\beta_{RS}\beta_{RC}}{\beta_{CS}}w^2
-\mu\left(1-\frac{\beta_{RS}}{\beta_{CS}}\right)w
+\mu\rho=0.
\]
Its discriminant is
\[
\Delta=\mu^2(\beta_{CS}-\beta_{RS})^2
-4\mu\rho\,\beta_{RS}\beta_{RC}\beta_{CS}.
\]
If \(\Delta\ge 0\), the roots are
\[
w_c^\pm
=
\frac{\mu(\beta_{CS}-\beta_{RS})\pm\sqrt{\Delta}}{2\beta_{RS}\beta_{RC}},
\]
and the corresponding \(y\)-coordinates are
\[
y_c^\pm
=
\frac{\beta_{CS}-\mu-(\beta_{CS}+\beta_{RC})w_c^\pm}{\beta_{CS}}.
\]
Feasibility requires
\[
0<w_c^\pm<\frac{\beta_{CS}-\mu}{\beta_{CS}+\beta_{RC}},
\]
which is equivalent to \(y_c^\pm>0\).

\medskip
\noindent\textbf{(iii) Stability of interior equilibria.}
At an interior equilibrium \((y^*,w^*)\), using the equilibrium identity
\[
\beta_{CS}(1-y^*-w^*)-\mu-\beta_{RC}w^*=0,
\]
the Jacobian simplifies to
\[
J(y^*,w^*)=
\begin{pmatrix}
-\beta_{CS}y^* & -(\beta_{CS}+\beta_{RC})y^*\\[4pt]
-\beta_{RS}w^* & -\beta_{RS}w^*-\mu\rho/w^*
\end{pmatrix}.
\]
Its trace is
\[
\operatorname{tr}J=-\beta_{CS}y^*-\beta_{RS}w^*-\frac{\mu\rho}{w^*}<0.
\]
Its determinant is
\[
\det J
=
\beta_{CS}y^*\left(\beta_{RS}w^*+\frac{\mu\rho}{w^*}\right)
-(\beta_{CS}+\beta_{RC})\beta_{RS}y^*w^*.
\]
Using the interior-equilibrium relation
\[
\mu\rho=\mu\left(1-\frac{\beta_{RS}}{\beta_{CS}}\right)w^*-\frac{\beta_{RS}\beta_{RC}}{\beta_{CS}}(w^*)^2,
\]
this simplifies to
\[
\det J
=
\frac{\beta_{CS}\mu y^*}{w^*}
\left(
2\rho-\left(1-\frac{\beta_{RS}}{\beta_{CS}}\right)w^*
\right).
\]
Since \(w_c^-<w_c^+\), one checks that \(\det J>0\) at \(w^*=w_c^-\) and \(\det J<0\) at \(w^*=w_c^+\). Therefore \(E_c^-\) is locally asymptotically stable, whereas \(E_c^+\) is a saddle.

\medskip
\noindent\textbf{(iv) Comparison of \(w_f\) and \(w_{\mathrm{thr}}\) when \(w_c^-<w_{\mathrm{thr}}<w_c^+\).}
We now show that if
\[
w_c^-<w_{\mathrm{thr}}<w_c^+,
\]
then necessarily
\[
w_f<w_{\mathrm{thr}},
\]
so the boundary equilibrium \(E_f\) is unstable in the \(y\)-direction.

On the boundary \(y=0\), the \(w\)-equation reduces to
\[
w'=\mu\rho+(\beta_{RS}-\mu)w-\beta_{RS}w^2=:G(w),
\]
whose unique positive root is \(w_f\).

For an interior equilibrium, the \(w\)-equation reduces to
\[
F(w)=\frac{\beta_{RS}\beta_{RC}}{\beta_{CS}}w^2
-\mu\left(1-\frac{\beta_{RS}}{\beta_{CS}}\right)w
+\mu\rho.
\]
Since \(F\) has roots \(w_c^-<w_c^+\), we may write
\[
F(w)=a(w-w_c^-)(w-w_c^+),
\qquad
a=\frac{\beta_{RS}\beta_{RC}}{\beta_{CS}}>0.
\]
Thus if \(w_c^-<w_{\mathrm{thr}}<w_c^+\), then
\[
F(w_{\mathrm{thr}})<0.
\]

Now, by the definition of \(w_{\mathrm{thr}}\),
\[
\beta_{CS}(1-w_{\mathrm{thr}})=\mu+\beta_{RC}w_{\mathrm{thr}}
\iff
1-w_{\mathrm{thr}}=\frac{\mu+\beta_{RC}w_{\mathrm{thr}}}{\beta_{CS}}.
\]
Hence
\[
\begin{aligned}
F(w_{\mathrm{thr}})
&=
\mu\rho-\mu w_{\mathrm{thr}}
+\frac{\beta_{RS}}{\beta_{CS}}w_{\mathrm{thr}}(\mu+\beta_{RC}w_{\mathrm{thr}})\\
&=
\mu\rho-\mu w_{\mathrm{thr}}+\beta_{RS}w_{\mathrm{thr}}(1-w_{\mathrm{thr}})\\
&=
G(w_{\mathrm{thr}}).
\end{aligned}
\]
Therefore \(G(w_{\mathrm{thr}})<0\). Since \(G(0)=\mu\rho>0\) and \(G\) is a concave-down quadratic with unique positive root \(w_f\), it follows that
\[
G(w)<0 \iff w>w_f.
\]
Hence \(G(w_{\mathrm{thr}})<0\) implies \(w_{\mathrm{thr}}>w_f\), as claimed.

\medskip
\noindent\textbf{(v) Absence of periodic orbits.}
We apply the Bendixson--Dulac criterion on
\[
\Omega^\circ=\{(y,w)\in\mathbb{R}_+^2:\ y>0,\ w>0,\ y+w<1\}.
\]
Choose the Dulac function
\[
B(y,w)=\frac{1}{yw}.
\]
Then
\[
Bf=\frac{1}{w}\Big(\beta_{CS}(1-y-w)-\mu-\beta_{RC}w\Big),
\]
so
\[
\frac{\partial}{\partial y}(Bf)=-\frac{\beta_{CS}}{w}.
\]
Also,
\[
Bg=\frac{\mu\rho}{yw}+\frac{\beta_{RS}}{y}(1-y-w)-\frac{\mu}{y},
\]
so
\[
\frac{\partial}{\partial w}(Bg)
=-\frac{\mu\rho}{yw^2}-\frac{\beta_{RS}}{y}.
\]
Therefore,
\[
\frac{\partial}{\partial y}(Bf)+\frac{\partial}{\partial w}(Bg)
=
-\frac{\beta_{CS}}{w}-\frac{\mu\rho}{yw^2}-\frac{\beta_{RS}}{y}<0
\qquad\text{for all }(y,w)\in\Omega^\circ.
\]
Hence the Dulac divergence has strict sign on the simply connected domain \(\Omega^\circ\). By the Bendixson--Dulac criterion, the scaled SUD-free subsystem admits no periodic orbits in \(\Omega^\circ\).

\medskip
\noindent\textbf{(vi). Global phase portrait.}
Since the subsystem is planar, positively invariant in \(\Omega\), and has no periodic orbits in \(\Omega^\circ\), every \(\omega\)-limit set in \(\Omega^\circ\) must be an equilibrium by the Poincar\'e--Bendixson theorem.
If
\[
0<w_c^-<\frac{\beta_{CS}-\mu}{\beta_{CS}+\beta_{RC}}<w_c^+,
\]
then the subsystem has a unique interior equilibrium \((y_c^-,w_c^-)\), which is locally asymptotically stable. Since no periodic orbits exist in \(\Omega^\circ\), this equilibrium is globally asymptotically stable in \(\Omega^\circ\).

If instead
\[
0<\frac{\beta_{CS}-\mu}{\beta_{CS}+\beta_{RC}}<w_c^-,
\]
then no interior equilibrium exists. In this case,
\[
F(w_{\mathrm{thr}})=a(w_{\mathrm{thr}}-w_c^-)(w_{\mathrm{thr}}-w_c^+)>0.
\]
Since \(F(w_{\mathrm{thr}})=G(w_{\mathrm{thr}})\), we obtain \(G(w_{\mathrm{thr}})>0\), which implies \(w_{\mathrm{thr}}<w_f\). Hence the boundary equilibrium \((0,w_f)\) is locally asymptotically stable. Again, because no periodic orbits exist, \((0,w_f)\) is globally asymptotically stable.

If \(0<w_c^\pm<\frac{\beta_{CS}-\mu}{\beta_{CS}+\beta_{RC}},\)
 then both \(E_c^-\) and \(E_c^+\) exist, then the equilibria in \(\Omega\) are the stable node \(E_f\), the stable node \(E_c^-\), and the saddle \(E_c^+\). Therefore every trajectory in \(\Omega^\circ\) converges to either \(E_f\) or \(E_c^-\). This proves bistability.

This completes the proof.
\end{proof}

%%%%%%%%%%%%%%%%%%%%%%%%%%%%%%%%%%%%%%%%%%%%%%%%%%%%%%%%%%%%%%%%%%%%%%%%%%%%%%%%%%%%
%%%%%%%%%%%%%%%%%%%%%%%%%%%%%%%%%%%%%
\subsection*{Proof of Theorem \ref{thm:yzw_stability}}
\begin{proof}
\textbf{Step 1. Jacobian at \(E_f\) and local stability.}
Write the reduced system in the form
\[
(y',z',w')=(F_1,F_2,F_3),
\qquad x=1-y-z-w.
\]
At \(E_f=(0,0,w_f)\), we have \(x_f=1-w_f\). Direct differentiation gives
\[
J(E_f)=
\begin{pmatrix}
\beta_{CS}(1-w_f)-\mu-\beta_{RC}w_f & \beta_{AS}(1-w_f) & 0\\
0 & -(\mu+\delta+\beta_{RA}w_f) & 0\\
-\beta_{RS}w_f & -(1-\phi)\delta w_f-\beta_{RS}w_f & \beta_{RS}(1-2w_f)-\mu
\end{pmatrix}.
\]
Because the second row has zeros off the diagonal, the eigenvalues are
\[
\lambda_y^{(f)}=\beta_{CS}(1-w_f)-\mu-\beta_{RC}w_f,\quad
\lambda_z^{(f)}=-(\mu+\delta+\beta_{RA}w_f),\quad
\lambda_w^{(f)}=\beta_{RS}(1-2w_f)-\mu.
\]
Moreover,
\[
\lambda_y^{(f)}
=
(\mu+\beta_{RC}w_f)
\left(
\frac{\beta_{CS}(1-w_f)}{\mu+\beta_{RC}w_f}-1
\right)
=
(\mu+\beta_{RC}w_f)(\mathcal{R}_y-1).
\]
Thus \(\lambda_y^{(f)}<0\) if and only if \(\mathcal{R}_y<1\). Since \(\lambda_z^{(f)}<0\) and \(\lambda_w^{(f)}<0\), it follows that \(E_f\) is locally asymptotically stable if and only if \(\mathcal{R}_y<1\).

\medskip
\textbf{Step 2. Jacobian at \(E_c^-\) and local stability.}
If \(0<w_c^-<w_{thr}\), then by Theorem~\ref{thm:yw}, \(E_c^-=(y_c^-,0,w_c^-)\) exists and is locally asymptotically stable on the invariant manifold \(z=0\). From
\[
z'=z(\beta_{AC}y-\mu-\delta-\beta_{RA}w)-(1-\phi)\delta z^2,
\]
we obtain
\[
\frac{\partial z'}{\partial y}(y,0,w)=0,\qquad
\frac{\partial z'}{\partial w}(y,0,w)=0,
\]
and
\[
\frac{\partial z'}{\partial z}(y,0,w)=\beta_{AC}y-(\mu+\delta+\beta_{RA}w).
\]
Therefore \(J(E_c^-)\) is block triangular:
\[
J(E_c^-)=
\begin{pmatrix}
J_{yw}(y_c^-,w_c^-) & *\\
0\ \ 0 & \lambda_z^{(c)}
\end{pmatrix},
\qquad
\lambda_z^{(c)}=\beta_{AC}y_c^--(\mu+\delta+\beta_{RA}w_c^-).
\]
Hence \(E_c^-\) is locally asymptotically stable in the full system if and only if
\[
\lambda_z^{(c)}<0
\iff
\mathcal{R}_z^{(c)}<1.
\]
\end{proof}

%%%%%%%%%%%%%%%%%%%%%%%%%%%%%%%%%%%%%%%%%%%%%%%

%%%%%%%%%%%%%%%%%%%%%%%%%%%%%%%%%%%%%%%%%%%%%%%%%%%%%%%%%%%%%%%%%%%%%%%%%%%%%%%%%%%%
\subsection*{Proof of Theorem \ref{thm:yzw_persistence}}
\begin{proof}
Define
\[
P(y,z,w):=yz.
\]
Then \(P>0\) on the interior of the state space and \(P=0\) on the parts of the boundary where either \(y=0\) or \(z=0\). Along interior solutions,
\[
\frac{d}{dt}\ln P=\frac{y'}{y}+\frac{z'}{z}.
\]

\smallskip
\emph{(i) Growth rate near the invariant set \(M_1:=\{y=0\}\cap\Omega\).}
At the substance-free equilibrium \(E_f=(0,0,w_f)\),
\[
\left.\frac{y'}{y}\right|_{E_f}
=
\beta_{CS}(1-w_f)-\mu-\beta_{RC}w_f
=
(\mu+\beta_{RC}w_f)(\mathcal{R}_y-1).
\]
Under the present hypotheses, \(0<w_c^-<w_{\mathrm{thr}}<w_c^+\) implies \(w_f<w_{\mathrm{thr}}\), so \(\mathcal R_y>1\). Hence the average growth rate of \(y\) near \(M_1\) is positive.

\smallskip
\emph{(ii) Growth rate near the invariant set \(M_2:=\{z=0\}\cap\Omega\).}
On \(M_2\), the dynamics reduce to the \((y,w)\)-subsystem. If
\[
0<w_c^-<w_{thr}<w_c^+,
\]
then every trajectory in \(M_2\) with \(y(0)>0\) converges to \(E_c^-=(y_c^-,0,w_c^-)\). Along \(M_2\),
\[
\frac{z'}{z}=\beta_{AC}y-(\mu+\delta+\beta_{RA}w).
\]
Taking the limit along trajectories converging to \(E_c^-\) gives
\[
\lim_{t\to\infty}\frac{z'}{z}
=
\beta_{AC}y_c^--(\mu+\delta+\beta_{RA}w_c^-)
=
(\mu+\delta+\beta_{RA}w_c^-)(\mathcal{R}_z^{(c)}-1).
\]
If \(\mathcal{R}_z^{(c)}>1\), this limit is strictly positive. Hence the average growth rate of \(z\) near \(M_2\) is positive.

\smallskip
\emph{(iii) Persistence conclusion.}
Combining (i) and (ii), every invariant set on the boundary that can attract interior solutions has strictly positive average Lyapunov exponent for \(\ln P\). By the standard average-Lyapunov persistence theorem, the classes \(y\) and \(z\) are uniformly persistent. That is, there exists \(\eta>0\) such that every solution with \(y(0)>0\) and \(z(0)>0\) satisfies
\[
\liminf_{t\to\infty}y(t)\ge\eta,
\qquad
\liminf_{t\to\infty}z(t)\ge\eta.
\]
This completes the proof.
\end{proof}
%%%%%%%%%%%%%%%%%%%%%%%%%%%%%%%%%%%%%%%%%%%%%%%%%%%%%%%%%%%%%%%%%%%%%%%%%%%%%%%%%%%%

\subsection*{Proof of Theorem \ref{thm:phi1_interior_yzw}}
\begin{proof}
Assume \(\phi=1\) and consider the \(yzw\) system~\eqref{eq:scaledphi10} on
\[
\Omega=\{(y,z,w)\in\mathbb R_+^3:\ y+z+w<1\},
\qquad x=1-y-z-w.
\]
Let \(E^*=(y^*,z^*,w^*)\in\Omega\) be an interior equilibrium, so that
\[
y^*>0,\qquad z^*>0,\qquad w^*>0.
\]

\medskip
\noindent\textbf{Step 1. Solving \(z'=0\).}
Since \(z^*>0\), the equation \(z'=0\) implies
\[
\beta_{AC}y^*-\mu-\delta-\beta_{RA}w^*=0,
\]
and hence
\begin{equation}\label{eq:y_of_w_phi1_clean}
y^*=\frac{\mu+\delta+\beta_{RA}w^*}{\beta_{AC}}.
\end{equation}

\medskip
\noindent\textbf{Step 2. Solving \(w'=0\).}
Since \(w^*>0\), the equation \(w'=0\) gives
\[
\mu\rho=w^*\big(\mu-\beta_{RS}x^*\big),
\]
which yields
\begin{equation}\label{eq:x_of_w_phi1_clean}
x^*=\frac{\mu(w^*-\rho)}{\beta_{RS}w^*}.
\end{equation}
Thus \(x^*>0\) requires \(w^*>\rho\).

\medskip
\noindent\textbf{Step 3. Expression for \(z^*\).}
Using \(x^*+y^*+z^*+w^*=1\), we obtain
\begin{equation}\label{eq:z_of_w_phi1_clean}
z^*=1-w^*-y^*-x^*.
\end{equation}

\medskip
\noindent\textbf{Step 4. Admissible interval for \(w^*\).}
Substituting \eqref{eq:y_of_w_phi1_clean} and \eqref{eq:x_of_w_phi1_clean} into \eqref{eq:z_of_w_phi1_clean}, define
\[
\Gamma(w):=wz(w)
=
-\left(1+\frac{\beta_{RA}}{\beta_{AC}}\right)w^2
+\left(1-\frac{\mu+\delta}{\beta_{AC}}-\frac{\mu}{\beta_{RS}}\right)w
+\frac{\mu\rho}{\beta_{RS}}.
\]
Since the quadratic coefficient is negative and
\[
\Gamma(0)=\frac{\mu\rho}{\beta_{RS}}>0,
\]
the quadratic \(\Gamma(w)\) has exactly one positive root and one negative root. Denote the positive root by \(w_z^+\), given in \eqref{eq:wzplus}. Then
\[
z^*>0 \iff 0<w^*<w_z^+.
\]
Combining this with \(x^*>0\) and \(w^*<1\), the admissible interval for interior equilibria is
\[
w^*\in\mathcal W=(\rho,\min\{1,w_z^+\}).
\]

\medskip
\noindent\textbf{Step 5. Reduction to a scalar equation for \(w^*\).}
At equilibrium, the \(y\)-equation can be written as
\[
0
=
y^*\Big[\beta_{CS}x^*-\mu-\beta_{RC}w^*-(\beta_{AC}+\beta_{AS})z^*\Big]
+z^*\Big[\beta_{AS}(x^*+y^*)+\beta_{RA}w^*\Big].
\]
Solving this equation for \(z^*\) gives
\[
z^*
=
\frac{y^*(\mu+\beta_{RC}w^*-\beta_{CS}x^*)}
{\beta_{AS}x^*-(\mu+\delta)}.
\]
Substituting \eqref{eq:y_of_w_phi1_clean} and \eqref{eq:x_of_w_phi1_clean} yields
\[
z^*
=
\frac{
\dfrac{\mu+\delta+\beta_{RA}w^*}{\beta_{AC}}
\left(\mu+\beta_{RC}w^*-\beta_{CS}\dfrac{\mu(w^*-\rho)}{\beta_{RS}w^*}\right)
}{
\beta_{AS}\dfrac{\mu(w^*-\rho)}{\beta_{RS}w^*}-(\mu+\delta)
}.
\]
Equating this with \eqref{eq:z_of_w_phi1_clean} and simplifying gives
\begin{equation}\label{eq:w_scalar_phi1_clean}
\frac{
[(\beta_{RS}-\mu)w^*+\mu\rho-\beta_{RS}(w^*)^2]
[(\beta_{AS}\mu-\mu\beta_{RS}-\delta\beta_{RS})w^*-\beta_{AS}\mu\rho]
}
{
\beta_{RS}w^*
\big[\beta_{RC}\beta_{RS}(w^*)^2
+(\beta_{AS}-\beta_{CS})\mu(w^*-\rho)
-\delta\beta_{RS}w^*\big]
}
=
\frac{\mu+\delta+\beta_{RA}w^*}{\beta_{AC}}.
\end{equation}
Therefore, any interior equilibrium must satisfy \eqref{eq:phi1_y_of_w} with \(w^*\in\mathcal W\), and conversely any \(w^*\in\mathcal W\) solving \eqref{eq:w_scalar_phi1_clean} generates an interior equilibrium.

\medskip
\noindent\textbf{Step 6. Extinction of \(z(t)\) when \(\beta_{AC}<\mu+\delta\).}
Since \(y(t)<1\) and \(w(t)\ge 0\), we have
\[
z'
=
z\big[\beta_{AC}y-\mu-\delta-\beta_{RA}w\big]
\le
z(\beta_{AC}-\mu-\delta).
\]
If \(\beta_{AC}<\mu+\delta\), then there exists \(\eta>0\) such that
\[
z'(t)\le -\eta z(t).
\]
By comparison,
\[
z(t)\le z(0)e^{-\eta t},
\]
and therefore
\[
\lim_{t\to\infty}z(t)=0.
\]
Hence no interior equilibrium can exist in this case.

\medskip
\noindent\textbf{Step 7. Existence of an interior equilibrium when \(\mathcal{R}_z^{(c)}>1\).}
If
\[
0<w_c^-<w_{thr}<w_c^+
\qquad\text{and}\qquad
\mathcal{R}_z^{(c)}>1,
\]
then by Theorem~\ref{thm:yzw_persistence}, the problematic-use class \(z\) is persistent in the reduced \(yzw\) system. Thus all components of the reduced \(yzw\) system, i.e., $y,z,w$, are persistent. This implies that the reduced \(yzw\) system, or the original SCAR model is permanent. By the standard permanent theory for dissipative systems (through the fixed point theory), permanence implies the existence of at least one interior equilibrium in \(\Omega\). This completes the proof.
\end{proof}

\subsection*{Proof of Corollary~\ref{cor:zy_to_zero}}
\begin{proof}
From the \(z\)-equation in \eqref{eq:scaled_reduced},
\[
z' = z\Big[\beta_{AC}y-\mu-\delta-(1-\phi)\delta z-\beta_{RA}w\Big].
\]
Since \(0\le y\le 1\), \(z\ge 0\), and \(w\ge 0\), we obtain
\[
z'
\le z(\beta_{AC}-\mu-\delta).
\]
If \(\beta_{AC}<\mu+\delta\), define
\[
\eta:=\mu+\delta-\beta_{AC}>0.
\]
Then
\[
z'(t)\le -\eta z(t),
\]
and by comparison,
\[
0\le z(t)\le z(0)e^{-\eta t}.
\]
Hence
\[
\lim_{t\to\infty} z(t)=0.
\]

Now assume in addition that \(w_{\mathrm{thr}}<w_c^-\). By Theorem~\ref{thm:yw}(iv)(c), the substance-free equilibrium \(E_f=(0,w_f)\) is globally asymptotically stable for the SUD-free subsystem
\begin{subequations}\label{eq:sudfree_cor_gen}
\begin{align}
y' &= y\Big[\beta_{CS}(1-y-w)-\mu-\beta_{RC}w\Big],\\
w' &= \mu\rho-w\Big[\mu-\beta_{RS}(1-y-w)\Big].
\end{align}
\end{subequations}

Since \(z(t)\to0\), the \((y,w)\)-dynamics of the full reduced system become asymptotically equivalent to \eqref{eq:sudfree_cor_gen}. Indeed, writing the \((y,w)\)-equations from \eqref{eq:scaled_reduced} as
\[
y'
=
y\Big[\beta_{CS}(1-y-w)-\mu-\beta_{RC}w\Big]
+G_1(y,z,w),
\]
\[
w'
=
\mu\rho-w\Big[\mu-\beta_{RS}(1-y-w)\Big]
+G_2(y,z,w),
\]
we have
\[
G_1(y,z,w)
=
-y(\beta_{CS}+\beta_{AC}+\beta_{AS})z
-(1-\phi)\delta yz
+z\Big[\beta_{AS}(1-z-w)+\beta_{RA}w\Big],
\]
\[
G_2(y,z,w)
=
-(1-\phi)\delta wz-\beta_{RS}wz.
\]
Because solutions remain in the positively invariant region
\[
\Omega=\{(y,z,w)\in\mathbb{R}_+^3:\ y+z+w\le 1\},
\]
all components are bounded, and therefore
\[
G_1(y(t),z(t),w(t))\to0,
\qquad
G_2(y(t),z(t),w(t))\to0
\qquad\text{as } t\to\infty.
\]
Thus the \((y,w)\)-dynamics are asymptotically governed by the SUD-free subsystem \eqref{eq:sudfree_cor_gen}, whose global attractor is \(E_f=(0,w_f)\). It follows that
\[
\lim_{t\to\infty}(y(t),w(t))=(0,w_f),
\]
and therefore
\[
\lim_{t\to\infty} y(t)=0,
\qquad
\lim_{t\to\infty} w(t)=w_f.
\]

Combining this with \(\lim_{t\to\infty} z(t)=0\), we obtain
\[
\lim_{t\to\infty}(y(t),z(t),w(t))=(0,0,w_f).
\]
This completes the proof.
\end{proof}
\section*{Declarations}

\noindent\textbf{Funding and/or Conflicts of Interest/Competing Interests and Availability of Data and Materials}

This work was partially supported by the National Institutes of Health (NIH) [National Institute of Alcohol Abuse and Alcoholism (R01AA031281)-Jinni Su]; a Jetstream2 AI Fellowship and an NVIDIA Academic Grant Award [Yixuan He]; and Simons Foundation for Mathematician [Yun Kang].

The authors declare that they have no known competing financial interests or personal relationships that could have appeared to influence the work reported in this paper.

Experiment datasets were not generated or analyzed during the current study.
\bibliography{references}
\end{document}